\documentclass[draft,narroweqnarray,inline]{ieee}
\pdfoutput=1
\usepackage{psfrag}
\usepackage{amsmath,amssymb}
\usepackage{synttree}
\usepackage{bbm}
\usepackage{graphicx}
\usepackage{alltt}
\usepackage{color_mre}
\definecolor{gris}{rgb}{0.85,0.85,0.85}

\def\classP{$\boldsymbol P$ }
\def\classNP{$\boldsymbol N \! \boldsymbol P$ }

\newtheorem{theorem}{Theorem}[section]
\newtheorem{definition}[theorem]{Definition}
 
\newtheorem{lemma}{{Lemma}}[section]

\newtheorem{property}{{Property}}[section]
\newtheorem{corollary}{{Corollary}}[section]
\def\len{\mbox{\em len}}
\def\combi{\mbox{\em combi}}
\def\pred{\mbox{\em pred}}

\def\succ{\mbox{\em succ}}

\newcommand{\comment}[1]{}

\begin{document}


\title[Hard 3-CNF-SAT problems are in \classP.]{A first step in proving `NP=P' :\linebreak 
Hard 3-CNF-SAT problems are in P}

\author[M.R\'emon & J.Barth\'elemy]{%
      Prof. Marcel R\'emon\authorinfo{%
M.R\'emon, Department of Mathematics,
      Namur University, Belgium;
        \mbox{marcel.remon@unamur.be}}
\and and Johan Barth\'elemy\authorinfo{%
J.Barth\'elemy, SMART Infrastructure Facility, University of Wollongong, Australia; \mbox{johan@uow.edu.au}} 
      }

\maketitle

\thispagestyle{empty}
\begin{abstract}
\noindent The relationship between the complexity classes \classP and
\classNP is an unsolved question in the field of theoretical computer
science. 
In the first part of this paper, a lattice framework is proposed to handle the 3-CNF-SAT problems, known to be in \classNP.  In the second section, we define a multi-linear descriptor function ${\cal H}_\varphi$ for any 3-CNF-SAT
problem $\varphi$ of size $n$, in the sense that ${\cal H}_\varphi : \{0,1\}^n
\rightarrow \{0,1\}^n$ is such that $Im \; {\cal H}_\varphi$ is the set
of all the solutions of $\varphi$.  A new ``merge" operation ${\cal H}_\varphi \bigwedge {\cal H}_{\psi}$ is defined, where $\psi$ is a single 3-CNF clause.
Given ${\cal H}_\varphi$ [but this can be of exponential complexity], the complexity needed for the
computation of $Im \; {\cal
  H}_\varphi$, the set of all solutions, is shown to be polynomial for ``hard" 3-CNF-SAT problems, i.e. the one with few ($\leq 2^k$) or no solutions. 
The third part uses the relation between ${\cal H}_\varphi$ and the indicator function $\mathbbm{1}_{{\cal S}_\varphi}$ for the set of solutions, to develop a {\em greedy} polynomial algorithm to solve ``hard" 3-CNF-SAT problems. 
\end{abstract}


\section*{Introduction}
\addcontentsline{toc}{section}{Introduction}
\section{Lattice framework for 3-CNF-SAT problems}
\subsection{The 3-CNF-SAT problem, a \classNP reference problem}
\noindent {\it Boolean formulae} are built in the usual way from
propositional variables $x_i$ and three logical connectives $\wedge$, $\vee$ and $\neg$, which are interpreted as conjunction, disjunction, and negation, respectively.  A {\it literal} is a propositional variable or the negation of a propositional variable, and a {\it clause} is a disjunction of literals.  A Boolean formula is {\it in conjunctive normal form} if and only if it is a conjunction of clauses. \\[12pt]
\noindent A {\it 3-CNF formula} $\varphi$ is a Boolean formula in
conjunctive normal form with exactly three literals per clause, like
$\varphi := (x_1 \vee x_2 \vee \neg x_3) \wedge (\neg x_2 \vee  x_3
\vee \neg x_4):= \psi_1 \wedge \psi_2 $.  A {\it 3-CNF formula} is
composed of $n$ propositional variables $x_i$ and $m$ clauses
$\psi_j$.  \\[12pt]
\noindent The {\it 3-CNF-satisfiability or 3-CNF-SAT problem} is to
decide whether there exists or not logical values for the
propositional variables, so that $\varphi$ can be true.  Until now, we
do not know whether it is possible or not to
check the satisfiability of any given {\it 3-CNF} formula $\varphi$ in
a polynomial time with respect of $n$, as the {\it 3-CNF-SAT} problem
is known to belong to the hardest problems in the class \classNP \hspace{-3pt}. See \cite{cormen2001} for details. 
\subsection{A matrix representation of the set of solutions for a 3-CNF formula}
\subsubsection{Definitions}
\noindent The {\it size} of a 3-CNF formula $\varphi$ is defined
as the size of the corresponding {\it Boolean circuit}, i.e. the
number of logical connectives in  $\varphi$.  Let us note
the following property :
\begin{eqnarray}
 \mbox{\it size($\varphi$)} = {\cal O}(m) = {\cal O}(\Delta \times n)
\label{no1}
\end{eqnarray}
\noindent where $\Delta = m / n $ is the {\it ratio} of
clauses with respect to variables. It seems that $\Delta \approx 4.258$ 
gives the most difficult 3-CNF-SAT problems.  See \cite{Crawford199631}. \\[12pt] 
\noindent Let $\varphi(x_1,x_2,\cdots,x_n)$ be a 3-CNF formula. 
 The set ${\cal S}_{\varphi}$ of all {\it satisfying} solutions is 
\begin{eqnarray}
 {\cal S}_{\varphi} = \{ (x_1,\cdots,x_n) \in
\{0,1\}^n \;  | \; \varphi(x_1,\cdots,x_n) =1 \} 
\label{no2}
\end{eqnarray}
Let $\Sigma_{\varphi} = \# \; {\cal S}_{\varphi}$ and $\bar{s}_1, \cdots,
\bar{s}_{\Sigma_{\varphi}}$ be the sorted [with respect to the binary order] elements of $ {\cal
  S}_{\varphi}$.  For $1 \leq j \leq \Sigma_{\varphi} : \bar{s}_j =
(s_j^1,\cdots,s_j^i,\cdots,s_j^n)$. We define the {\it ${\cal S}_\varphi$-matrix} representation of ${\cal S}_{\varphi} $ as
$[{\cal S}_{\varphi}]$ :
\begin{eqnarray}
[{\cal S}_{\varphi}] =  \left( \begin{array}{ccc}
x_1 & x_i & x_n \\
\hline
s_1^1 & \cdots & s_1^n \\
\vdots & s_j^i & \vdots  \\
s_{\Sigma_{\varphi}}^1& \cdots & s_{\Sigma_{\varphi}}^n \end{array}
\right)
\label{no3}
\end{eqnarray}
\subsubsection{Examples}
\noindent The set of solutions for any single clause $\psi_i$ will be represented by a $7 \times 3$
matrix. For example,
\begin{eqnarray*}
[{\cal S}_{\psi_1}] =  [{\cal S}_{x_1 \vee x_2 \vee \neg x_3}] =
 \left( \begin{array}{ccc}
x_1 & x_2 & x_3 \\
\hline
0 & 0 & 0 \\
0 & 1 & 0 \\
0 & 1 & 1 \\
1 & 0 & 0 \\
1 & 0 & 1 \\
1 & 1 & 0 \\
1 & 1 & 1 \\
\end{array} \right) \mbox{  and  } 
 [{\cal S}_{\psi_2}] = [{\cal S}_{\neg x_2 \vee  x_3
\vee \neg x_4}] =
 \left( \begin{array}{ccc}
x_2 & x_3 & x_4 \\
\hline
0 & 0 & 0 \\
0 & 0 & 1 \\
0 & 1 & 0 \\
0 & 1 & 1 \\
1 & 0 & 0 \\
1 & 1 & 0 \\
1 & 1 & 1 \\
\end{array} \right)  
\end{eqnarray*}
\mbox{}\\[12pt]
${\cal S}_{\psi_1 \wedge \psi_2}$ will be represented by a $12
\times 4$ matrix :
\begin{eqnarray*}
 [{\cal S}_{\psi_1 \wedge \psi_2}] = [{\cal S}_{ (x_1 \vee x_2 \vee \neg x_3) \wedge (\neg x_2 \vee  x_3
\vee \neg x_4)}] =
 \left( \begin{array}{cccc}
x_1 & x_2 & x_3 & x_4\\
\hline
0 & 0 & 0 & 0 \\
0 & 0 & 0 & 1 \\
0 & 1 & 0 & 0 \\
0 & 1 & 1 & 0 \\
0 & 1 & 1 & 1 \\
1 & 0 & 0 & 0 \\
1 & 0 & 0 & 1 \\
1 & 0 & 1 & 0 \\
1 & 0 & 1 & 1 \\
1 & 1 & 0 & 0 \\
1 & 1 & 1 & 0 \\
1 & 1 & 1 & 1 \\
\end{array} \right)  
\end{eqnarray*}
\mbox{}
\subsection{First properties for ${\cal S}_\varphi$-matrices}
\subsubsection{Extension to new variables}
\noindent Let $A$ be a ${\cal S}_\varphi$-matrix, $A$ can be {\it extended} to new
propositional variables by adding columns filled with the neutral sign
``.'', meaning that the corresponding variable can be set either to 0 or 1. This new matrix $\overline{A}$ is equivalent to $A$.
\begin{eqnarray}
 A = 
 \left( \begin{array}{ccc}
x_1 & x_2 & x_4\\
\hline
a_1^1 & a_1^2 & a_1^4 \\ 
a_j^1 & a_j^2 & a_j^4 \\ 
a_{\Sigma_{\varphi}}^1 & a_{\Sigma_{\varphi}}^2 & a_{\Sigma_{\varphi}}^4 \\ 
\end{array} \right)
\equiv 
\left( \begin{array}{cccc}
x_1 & x_2 & x_3 & x_4\\
\hline
a_1^1 & a_1^2 & . [_1^0] & a_1^4 \\ 
a_j^1 & a_j^2 & . & a_j^4 \\ 
a_{\Sigma_{\varphi}}^1 & a_{\Sigma_{\varphi}}^2 & . & a_{\Sigma_{\varphi}}^4 \\ 
\end{array} \right)=  \overline{A}  \label{extension}
\end{eqnarray}
\mbox{}
\subsubsection{The {\bf join} operation of ${\cal S}_\varphi$-matrices}
\noindent Let $A$ and $B$ be two ${\cal S}_\varphi$-matrices and $\{x_1, \cdots, x_n\}$ the
union of their support variables.  Let $\overline{A}$ and $\overline{B}$ be their
extensions over $\{x_1, \cdots, x_n\}$.  Then we define the {\it
  join operation} of $A$ and $B$ by
\begin{eqnarray}
 A \vee B = 
 \left( \begin{array}{c}
x_1 \; \cdots \; \; x_n\\
\hline
\overline{A} \\ 
\overline{B} \\ 
\end{array} \right)
\end{eqnarray}
Of course, this new matrix should be reordered so that the lines are
in a ascending binary order, which can yield sometimes in replacing a line with
a neutral sign by two lines with a one and a zero.
\subsubsection{The {\bf meet} operation of ${\cal S}_\varphi$-matrices}
\noindent Let $A$ and $B$ be two ${\cal S}_\varphi$-matrices,
 $\overline{A}$ and $\overline{B}$
their extensions to the joint set of propositional variables. 
Let $\overline{A}_k$ and $\overline{B}_l$ be the {\it one line matrices}
such that :
\begin{eqnarray}
\overline{A} =
\underset{k=1}{\overset{\Sigma_{\overline{A}}}{\bigvee}}
\overline{A}_k \mbox{ and }
\overline{B} =
\underset{l=1}{\overset{\Sigma_{\overline{B}}}{\bigvee}}
\overline{B}_l 
\end{eqnarray}
 We
define {\it the meet operation} of $A$ and $B$ as 
\begin{eqnarray}
 A \wedge B \equiv \overline{A} \wedge \overline{B} = 
\left( \underset{k=1}{\overset{\Sigma_{\overline{A}}}{\bigvee}}
\overline{A}_k \right) \wedge 
\left( 
\underset{l=1}{\overset{\Sigma_{\overline{B}}}{\bigvee}}
\overline{B}_l \right)=
\underset{k=1}{\overset{\Sigma_{\overline{A}}}{\bigvee}} \;
\underset{l=1}{\overset{\Sigma_{\overline{B}}}{\bigvee}}
\left( \overline{A}_k \wedge \overline{B}_l \right) = 
\underset{k=1}{\overset{\Sigma_{\overline{A}}}{\bigvee}} \;
\underset{l=1}{\overset{\Sigma_{\overline{B}}}{\bigvee}}
\overline{C}_{k,l} 
\end{eqnarray}
where 
\begin{eqnarray}
\overline{C}_{k,l} = 
\left( \begin{array}{ccc}
x_1 & x_i & x_n \\
\hline
a_k^1 & a_k^i & a_k^n \\
\end{array}
\right) \wedge 
\left( \begin{array}{ccc}
x_1 & x_i & x_n \\
\hline
b_l^1 & b_l^i & b_l^n \\
\end{array}
\right)
= \left\{ \begin{array}{l}
\; \emptyset \mbox{ if } \exists \; c_m^i = \mbox{\it ``NaN"}\\
\left( \begin{array}{ccc}
x_1 & x_i & x_n \\
\hline
c_m^1 & c_m^i & c_m^n \\
\end{array}
\right) \mbox{ otherwise }
\end{array} \right.
\end{eqnarray}
with
\begin{eqnarray}
c_m^i = 
\left\{ \begin{array}{l}
a_k^i \mbox{ if } a_k^i = b_l^i \\
a_k^i \mbox{ if } b_l^i = ``\cdot"  \\
b_l^i \mbox{ if } a_k^i = ``\cdot"  \\
\mbox{\it ``NaN"} \mbox{ otherwise} 
\end{array} \right.
\end{eqnarray}
\mbox{}
\subsubsection{The empty and full ${\cal S}_\varphi$-matrices}
\noindent Let us call $\emptyset$, the {\it empty matrix}, with no
line at all.  The empty matrix is neutral for the join operator
$\vee$ and absorbing for the meet operator $\wedge$.\\[12pt]
\noindent Let us define $\Omega$, the {\it full matrix}, as a one line matrix
with only neutral signs $``\cdot"$ in it. The full matrix is neutral for $\wedge$ and absorbing for $\vee$.
\mbox{}
\subsubsection{Lattice structure of ${\cal S}_\varphi$-matrices}
\noindent A {\it semi-lattice} $(X,\vee)$ is a pair consisting of a set X
and a binary operation $\vee$ which is associative, commutative, and
idempotent. \\[12pt]
\noindent Let us note ${\cal A}$ the set of all the
${\cal S}_\varphi$-matrices. Then $({\cal A},\vee)$ and $({\cal A},\wedge)$ are
both semi-lattices, respectively called {\it join} and {\it meet}
semi-lattices. \\[12pt] 
\noindent Let us define {\it the two absorption laws} as  $x = x \vee
(x \wedge y)$  and its dual $x = x \wedge (x \vee y)$.
A {\it lattice} is an algebra $(X, \vee, \wedge)$ satisfying equations
expressing associativity, commutativity, and idempotence of $\vee$ and
$\wedge$, and satisfying the two absorption equations. \\[12pt]
\noindent Therefore, {\bf \em $({\cal A},\vee, \wedge)$ is a lattice} over the set of
${\cal S}_\varphi$-matrices with respect to the join and meet
operators. Indeed, ${\cal S}_\varphi$-matrices satisfy the absorption equations as 
${\cal S}_\varphi = {\cal S}_{\varphi \vee ( \varphi \wedge \varphi')} = {\cal S}_{\varphi \wedge ( \varphi \vee \varphi')}$. \\[12pt]
Moreover, $({\cal A},\vee, \wedge)$ is {\bf \em a
  distributive bounded
lattice} as $\wedge$ is distributive with respect to $\vee$ and $A
\vee \Omega = \Omega \; \; \& \;\; A \wedge \emptyset = \emptyset \;\; \forall A
\in {\cal A}$.  See \cite{burris2012a} for more details over lattices. 
\subsection{``Hard'' 3-CNF-SAT problems}
\begin{definition}
A ``hard'' 3-CNF-SAT problem $\varphi$ is defined in this paper as a problem
with a small or limited set of solutions, in the sense that the number
of solutions is bounded :
\begin{eqnarray}
\Sigma_\varphi = 2^{k}   =
2^{ {\cal O}(1)} \; \mbox{ [for some $k$]} \label{hard}
\end{eqnarray}
Note : the problem is said to be ``hard'' in the sense that the probability to get a
solution at random $[= \frac{\Sigma_\varphi}{2^n}]$ tends to zero as
$n$ tends to infinity. The hardiest 3-CNF-SAT problems are the one
without solution. This paper only considers ``hard" 3-CNF-SAT problems.
\end{definition}
\section{The multi-linear descriptor function  ${\cal H}_\varphi$}
\subsection{Characterization theorem of ${\cal S}_\varphi$ via the descriptor function ${\cal H}_\varphi$}
\begin{theorem}\label{thm_caracterization}{\bf Every non empty ${\cal S}_\varphi$-matrix of $n$ literals  can be
    characterized by a single $n$-dimensional descriptor function ${\cal H}_\varphi$ : $\{0,1\}^n
\rightarrow \{0,1\}^n$ such that $Im \; {\cal H}_\varphi = {\cal S}_\varphi$.} 
\begin{eqnarray}
\forall \;  [{\cal S}_\varphi] &=&  \left( \begin{array}{ccc}
x_1 & x_i & x_n \\
\hline
s_1^1 & \cdots & s_1^n \\
\vdots & s_j^i & \vdots  \\
s_{\Sigma_{\varphi}}^1& \cdots & s_{\Sigma_{\varphi}}^n \end{array}
\right)\neq \emptyset
\; , \; \exists \; n \mbox{ functions } h_i : \{0,1\}^i
\rightarrow \{0,1\} \mbox{ such that } \nonumber \\[12pt]
\mbox{}
[{\cal S}_\varphi] &=&  \underset{(\alpha_1,\cdots,\alpha_n) \in \{0,1\}^n}{\bigvee} \!\!\!
 \left( \begin{array}{ccccc}
x_1 & \cdots & x_i & \cdots & x_n \\
\hline
h_1(\alpha_1) & \cdots & h_i(\alpha_1,\cdots,\alpha_i) & \cdots &
h_n(\alpha_1, \cdots, \alpha_n) \\
\end{array}
\right) \label{h-def}\\
\mbox{} \nonumber\\
&\stackrel{notation}{\equiv}& \left[ \begin{array}{ccc}
h_1(\alpha_1) &
\hdots &
h_n(\alpha_1,\cdots,\alpha_n)
\end{array}
\right]\stackrel{notation}{\equiv} \left[ \begin{array}{c}
{\cal H}_\varphi(\alpha_1,\cdots,\alpha_n)
\end{array}
\right] \label{no12} 
\end{eqnarray}
\end{theorem}
\mbox{}\\[12pt]
\noindent So, the knowledge of ${\cal H}_\varphi(\alpha_1, \cdots,
\alpha_n) $ characterizes fully $ [{\cal S}_\varphi] $.  ${\cal H}_\varphi(\alpha_1, \cdots, \alpha_n)$ is
called {\bf the descriptor function of ${\cal S}_\varphi$.}  {\em All operations are done in a {\bf \em mod(2)} framework.}  Before
proving the existence of such a function, let us consider some
examples.
\mbox{}\\[12pt]
\noindent {\it Examples of ${\cal H}_\varphi$ : } 
\begin{flalign}
\bullet  \hspace{1cm} \varphi & = (\neg x_1 \vee \neg x_2 \vee \neg x_3) \label{example1}\\
 [{\cal S}_\varphi]& = \hspace{-0.7cm}  \underset{(\alpha_1,\cdots,\alpha_3) \in
  \{0,1\}^3}{\bigvee}  
 \left( \begin{array}{ccc}
x_1 & x_2 & x_3  \\
\hline
\alpha_1 & \alpha_2 & \alpha_1 \alpha_2 \alpha_3+\alpha_3 \\
\end{array}
\right)_{\mbox{\tiny (mod 2)}} \hspace{-0.8cm} \equiv \left[ \begin{array}{ccc}
\alpha_1 &
\alpha_2 &
\alpha_1 \alpha_2 \alpha_3 + \alpha_3
\end{array}
\right] \equiv \left[ \begin{array}{c}
{\cal H}_\varphi
\end{array}
\right] \nonumber \\
 [{\cal S}_\varphi]& = 
 \left[ {\tiny \begin{array}{ccccc}
&\alpha_1 & \alpha_2 & \alpha_3  \\
\hline
{\cal H}_\varphi&(\; 0 & 0 & 0 \;) \\
{\cal H}_\varphi&(\; 0 & 0 & 1 \;) \\
{\cal H}_\varphi&(\; 0 & 1 & 0 \;) \\
{\cal H}_\varphi&(\; 0 & 1 & 1 \;) \\
{\cal H}_\varphi&(\; 1 & 0 & 0 \;) \\
{\cal H}_\varphi&(\; 1 & 0 & 1 \;) \\
{\cal H}_\varphi&(\; 1 & 1 & 0 \;) \\
{\cal H}_\varphi&(\; 1 & 1 & 1 \;) \\
\end{array}
}
\right]  =  \left[ {\tiny \begin{array}{ccc}
x_1 & x_2 & x_3 \\
\hline
0 & 0 & 0 \\
0 & 0 & 1 \\
0 & 1 & 0 \\
0 & 1 & 1 \\
1 & 0 & 0 \\
1 & 0 & 1 \\
\fcolorbox{gris}{gris}{1} & \fcolorbox{gris}{gris}{1} & \fcolorbox{gris}{gris}{0} \\
\fcolorbox{gris}{gris}{1} & \fcolorbox{gris}{gris}{1} & \fcolorbox{black}{gris}{0} \\
\end{array}
}
\right]  =  \left[ {\tiny \begin{array}{cccc}
x_1 & x_2 & x_3 \\
\hline
0 & 0 & 0 \\
0 & 0 & 1 \\
0 & 1 & 0 \\
0 & 1 & 1 \\
1 & 0 & 0 \\
1 & 0 & 1 \\
\fcolorbox{gris}{gris}{1} & \fcolorbox{gris}{gris}{1} & \fcolorbox{gris}{gris}{0} \\
\end{array}
}
\right] \nonumber
\end{flalign}
\mbox{}
\begin{flalign}
\bullet  \hspace{1cm} \varphi & = (x_1 \vee x_2 \vee \neg x_3) \wedge (\neg x_2 \vee  x_3
\vee \neg x_4) \label{example2}\\
 [{\cal S}_\varphi]& = \hspace{-0.7cm} \underset{(\alpha_1,\cdots,\alpha_4) \in
  \{0,1\}^4}{\bigvee}  
 \left( \begin{array}{cccc}
x_1 & x_2 & x_3 & x_4 \\
\hline
\alpha_1 & \alpha_2 & (\alpha_1+1)(\alpha_2+1)\alpha_3+\alpha_3 &
\alpha_2 (\alpha_3+1) \alpha_4 +\alpha_4\\
\end{array}
\right)_{\mbox{\tiny (mod 2)}} \nonumber \\[12pt]
& \equiv  \left[ \begin{array}{cccc}
\alpha_1 & 
\alpha_2 & 
(\alpha_1+1)(\alpha_2+1)\alpha_3+\alpha_3 &
\alpha_2 (\alpha_3+1) \alpha_4 +\alpha_4
\end{array}
\right]  \equiv \left[ \begin{array}{c}
{\cal H}_\varphi
\end{array}
\right] \nonumber \\[12pt] 
& =  
 \left[ {\tiny \begin{array}{ccccc}
& \alpha_1 & \alpha_2 & \alpha_3 & \alpha_4 \\
\hline
{\cal H}_\varphi & (\; 0 & 0 & 0 & 0 \;) \\
{\cal H}_\varphi & (\; 0 & 0 & 0 & 1 \;) \\
{\cal H}_\varphi & (\; 0 & 0 & 1 & 0 \;) \\
{\cal H}_\varphi & (\; 0 & 0 & 1 & 1 \;) \\
{\cal H}_\varphi & (\; 0 & 1 & 0 & 0 \;) \\
{\cal H}_\varphi & (\; 0 & 1 & 0 & 1 \;) \\
{\cal H}_\varphi & (\; 0 & 1 & 1 & 0 \;) \\
{\cal H}_\varphi & (\; 0 & 1 & 1 & 1 \;) \\
{\cal H}_\varphi & (\; 1 & 0 & 0 & 0 \;) \\
{\cal H}_\varphi & (\; 1 & 0 & 0 & 1 \;) \\
{\cal H}_\varphi & (\; 1 & 0 & 1 & 0 \;) \\
{\cal H}_\varphi & (\; 1 & 0 & 1 & 1 \;) \\
{\cal H}_\varphi & (\; 1 & 1 & 0 & 0 \;) \\
{\cal H}_\varphi & (\; 1 & 1 & 0 & 1 \;) \\
{\cal H}_\varphi & (\; 1 & 1 & 1 & 0 \;) \\
{\cal H}_\varphi & (\; 1 & 1 & 1 & 1 \;) \\
\end{array}
}
\right]  =  \left[ {\tiny \begin{array}{cccc}
x_1 & x_2 & x_3 & x_4 \\
\hline
0 & 0 & 0 & 0 \\
0 & 0 & 0 & 1 \\
\fcolorbox{gris}{gris}{0} & \fcolorbox{gris}{gris}{0} & \fcolorbox{black}{gris}{0} & \fcolorbox{gris}{gris}{0} \\
\fcolorbox{gris}{gris}{0} & \fcolorbox{gris}{gris}{0} & \fcolorbox{black}{gris}{0} & \fcolorbox{gris}{gris}{1} \\
0 & 1 & 0 & 0 \\
\fcolorbox{gris}{gris}{0} & \fcolorbox{gris}{gris}{1} & \fcolorbox{gris}{gris}{0} & \fcolorbox{black}{gris}{0} \\
0 & 1 & 1 & 0 \\
0 & 1 & 1 & 1 \\
1 & 0 & 0 & 0 \\
1 & 0 & 0 & 1 \\
1 & 0 & 1 & 0 \\
1 & 0 & 1 & 1 \\
1 & 1 & 0 & 0 \\
\fcolorbox{gris}{gris}{1} & \fcolorbox{gris}{gris}{1} & \fcolorbox{gris}{gris}{0} & \fcolorbox{black}{gris}{0} \\
1 & 1 & 1 & 0 \\
1 & 1 & 1 & 1 \\
\end{array}
}
\right]  =  \left[ {\tiny \begin{array}{cccc}
x_1 & x_2 & x_3 & x_4 \\
\hline
\fcolorbox{gris}{gris}{0} & \fcolorbox{gris}{gris}{0} & \fcolorbox{gris}{gris}{0} & \fcolorbox{gris}{gris}{0} \\
\fcolorbox{gris}{gris}{0} & \fcolorbox{gris}{gris}{0} & \fcolorbox{gris}{gris}{0} & \fcolorbox{gris}{gris}{1} \\
\fcolorbox{gris}{gris}{0} & \fcolorbox{gris}{gris}{1} & \fcolorbox{gris}{gris}{0} & \fcolorbox{gris}{gris}{0} \\
0 & 1 & 1 & 0 \\
0 & 1 & 1 & 1 \\
1 & 0 & 0 & 0 \\
1 & 0 & 0 & 1 \\
1 & 0 & 1 & 0 \\
1 & 0 & 1 & 1 \\
\fcolorbox{gris}{gris}{1} & \fcolorbox{gris}{gris}{1} & \fcolorbox{gris}{gris}{0} & \fcolorbox{gris}{gris}{0} \\
1 & 1 & 1 & 0 \\
1 & 1 & 1 & 1 \\
\end{array}
}
\right] \nonumber
\end{flalign}
\mbox{}\\[24pt]
\newpage
\begin{proof} (Existence of ${\cal H}_\varphi$) \hspace{24pt} [$\;$ Remember : all operations in {\em mod(2)}$\;$ ]
\mbox{}\\[12pt]
\noindent $\bullet$ The theorem is satisfied for $n=1$ as 
\begin{eqnarray*}
\left[ \begin{array}{c}
x_1 \\
\hline
1  \\
\end{array} \right]
= 
\left( \begin{array}{c}
x_1 \\
\hline
h_1(\alpha_1) \equiv 1  \\
\end{array} \right) \; ; \; 
\left[ \begin{array}{c}
x_1 \\
\hline
0  \\
\end{array} \right]
= 
\left( \begin{array}{c}
x_1 \\
\hline
h_1(\alpha_1) \equiv 0  \\
\end{array} \right) \; ; \; 
\left[ \begin{array}{c}
x_1 \\
\hline
0  \\
1 \\
\end{array} \right]
= 
\underset{\alpha_1 \in
  \{0,1\}}{\bigvee}
 \left( \begin{array}{c}
x_1 \\
\hline
\alpha_1   \\ 
\end{array} \right)
\end{eqnarray*}
$\bullet$
Let the theorem be true for $n-1$ and $[{\cal S}]$ be a ${\cal S}_\varphi$-matrix of
dimension $n$.  There exist two ${\cal S}_\varphi$-matrices $[{\cal S}_1]$ and
$[{\cal S}_2]$ of size $n-1$ such that :
\begin{flalign*}
[{\cal S}] =& \left[ \begin{array}{c|c}
x_1 & x_2 \cdots x_n \\
\hline
0 & [{\cal S}_1]   \\ 
\end{array} \right]
\vee \left[ \begin{array}{c|c}
x_1 & x_2 \cdots x_n \\
\hline
1 & [{\cal S}_2]   \\
\end{array} \right] \\
& \hspace{-12pt} \mbox{as  $[{\cal S}]$ can be divided in two sets of lines, the ones beginning with $0$ and the ones with $1$.}   \\
& \hspace{-12pt} \mbox{Using the recurrence hypothesis :} \\
[{\cal S}] =& \underset{\alpha_i \in
  \{0,1\}}{\bigvee} \left( \begin{array}{c|c}
x_1 & x_2 \cdots x_n \\
\hline
0 & f_2(\alpha_2) \cdots f_n(\alpha_2,\cdots,\alpha_n)   \\ 
\end{array} \right)
\underset{\alpha_i \in
  \{0,1\}}{\bigvee} \left( \begin{array}{c|c}
x_1 & x_2 \cdots x_n \\
\hline
1 & g_2(\alpha_2) \cdots g_n(\alpha_2,\cdots,\alpha_n)   \\
\end{array} \right) \\
\end{flalign*}
\begin{eqnarray*}
\mbox{Thus \hspace{1cm}}  [{\cal S}] = \underset{\alpha_i \in
  \{0,1\}}{\bigvee} \left( \begin{array}{c}
 x_1 \cdots x_n \\
\hline
h_1(\alpha_1) \cdots h_n(\alpha_1,\cdots,\alpha_n)   \\ 
\end{array} \right) \mbox{\hspace{2,8cm}}
\end{eqnarray*}
\mbox{\hspace{1cm}} where 
\begin{eqnarray*}
 h_1(\alpha_1) & = & \alpha_1 \\
h_i(\alpha_1,\cdots,\alpha_i) & = &  \underbrace{(\alpha_1+1) f_i(\alpha_2,\cdots,\alpha_i)}_\text{for lines where $x_1 = 0$} \;\;+ \underbrace{\alpha_1
g_i(\alpha_2,\cdots,\alpha_i)}_\text{for lines where $x_1 = 1$} \; \;  \; \; \mbox{ for } i \neq 1 
\end{eqnarray*}
\end{proof}
\begin{definition}{\bf Length of ${\cal H}_\varphi$ and $h_i(\alpha_1,\cdots,\alpha_i)$.}\\[12pt]
\noindent Let $\len(h_i)$ be defined as the
  {\bf number of terms} in $h_i(\alpha_1,\cdots,\alpha_i).$ \label{len-def} \\
\noindent Let $\len({\cal H}_\varphi)$ be defined as the {\bf maximum length} of  $h_i(\alpha_1,\cdots,\alpha_i) : \len({\cal H}_\varphi) = \max_i \len(h_i)$
\end{definition}
\mbox{} \\
\vspace{-12pt}
\begin{corollary}{\bf The descriptor function ${\cal H}_\varphi$ is a $n$-dimensional modulo-2 multi-linear  combination of $\alpha_i$.}
\end{corollary}
\begin{proof}
This is a mere consequence of the definition of
$h_i(\alpha_1,\cdots,\alpha_i)$ in {\em Theorem
  \ref{thm_caracterization}}. \end{proof} 
\mbox{} \\
\vspace{-12pt}
\begin{corollary} Let $A \subseteq \{\alpha_1, \cdots,
\alpha_n\}$, with $\alpha_i \in \{0,1\}$ and 
$h_i \in \combi(A) \stackrel{notation}{\Leftrightarrow} h_i \;$ is a multi-linear combination
  of $\alpha_i \in A$, modulo 2, then 
 \[h_i(\alpha_1, \cdots, \alpha_i) \in \combi(\{\alpha_1, \cdots, \alpha_i\}) \mbox{ , } \;\; \len(h_i) \leq 2^{i} \;\; \mbox{ and } \;\;  \len({\cal H}_\varphi) \leq 2^n \]
\end{corollary}
\mbox{}\\
{\em Example :}
\vspace{-12pt}
\begin{eqnarray*}
h_2(\alpha_1,\alpha_2) \in \combi(\{\alpha_1,\alpha_2\}) & \Rightarrow &
h_2(\alpha_1,\alpha_2) = \delta_{(0,0)} \alpha_1^0 \alpha_2^0 + \delta_{(1,0)} 
\alpha_1 + \delta_{(0,1)} \alpha_2 + \delta_{(1,1)} \alpha_1 \alpha_2 \\
& \Rightarrow & \len(h_2) \leq 2^{(\# \{\alpha_1,\alpha_2\})} = 2^2 \\
\end{eqnarray*}

\subsection{Computation of ${\cal H}_\varphi$}
\begin{theorem}{\bf Simple characterization theorem (one-clause 3-CNF formula)} \\[12pt] 
Consider the 3-CNF formula, consisting of only one clause $\psi \equiv [\neg] x_r
\vee [\neg] x_s \vee [\neg] x_t$ where $1 \leq r < s < t \leq
n$. 
$[{\cal S}_\psi]$ can be characterized by the
following $[ {\cal H}_\psi ] \equiv [ h_i(\alpha_1, \cdots, \alpha_i) ]$
descriptor function where : 
\begin{eqnarray}
h_i(\alpha_1, \cdots, \alpha_i) & = & \alpha_i \;\;\; \forall \; i
\neq t \; (1 \leq i \leq n) 
\nonumber \\
h_t(\alpha_r,\alpha_s,\alpha_t) & = & \left\{
\begin{array}{ll}
(\alpha_{r}+1)(\alpha_s+1)(\alpha_t+1)+\alpha_t \;\;  &
\mbox{if} \;\; \psi = x_r \vee x_s \vee x_t \\
(\alpha_{r}+1)(\alpha_s+1)\; \alpha_t+\alpha_t \;\;  &
\mbox{if} \;\; \psi = x_r \vee x_s \vee \neg x_t \\
(\alpha_{r}+1)\; \alpha_s \; (\alpha_t+1)+\alpha_t \;\;  &
\mbox{if} \;\; \psi = x_r \vee \neg  x_s \vee x_t \\
(\alpha_{r}+1)\; \alpha_s \; \alpha_t +\alpha_t \;\;  &
\mbox{if} \;\; \psi = x_r \vee \neg x_s \vee \neg x_t \\
\alpha_{r} \; (\alpha_s+1)(\alpha_t+1)+\alpha_t \;\;  &
\mbox{if} \;\; \psi = \neg x_r \vee x_s \vee x_t \\
\alpha_{r}\;(\alpha_s+1)\; \alpha_t+\alpha_t \;\;  &
\mbox{if} \;\; \psi = \neg x_r \vee x_s \vee \neg x_t \\
\alpha_{r}\; \alpha_s \; (\alpha_t+1)+\alpha_t \;\;  &
\mbox{if} \;\; \psi = \neg x_r \vee \neg  x_s \vee x_t \\
\alpha_{r}\; \alpha_s \; \alpha_t +\alpha_t \;\;  &
\mbox{if} \;\; \psi = \neg x_r \vee \neg x_s \vee \neg x_t \\
\end{array} \right. \label{def_h} \\
\nonumber
\end{eqnarray}
\end{theorem}
\begin{proof}
The proof is straightforward. See (\ref{example1}) for an example. \\
\end{proof}
\newpage
\begin{theorem}{\bf General descriptor function theorem} \label{algorithm} \\[12pt] 
The descriptor function ${\cal H}_{\varphi \wedge \psi}(\alpha_1, \cdots, \alpha_n), \;$ for the conjunction of a
  3-CNF formulae $\varphi$ with ${\cal F}_\varphi$ as descriptor function and a 3-CNF clause $\psi$ associated
to ${\cal G}_\psi$ can be
computed via a general algorithm. \\  
Let $\Lambda = \{ (\alpha_1,
\cdots, \alpha_n) \in \{0,1\}^n \; : {\cal F}_\varphi(\alpha_1, \cdots, \alpha_n) =
{\cal G}_\psi(\alpha_1, \cdots, \alpha_n) \}$.
Then the {\bf following algorithm} will give the exact ${\cal
  H}_{\varphi \wedge \psi}(\alpha_1, \cdots, \alpha_n) \;$ :
\begin{eqnarray*}
\forall (\alpha_1, \cdots, \alpha_n) \in \Lambda \; : \; {\cal H}_{\varphi \wedge \psi}(\alpha_1, \cdots, \alpha_n) &:=&{\cal F}_\psi(\alpha_1, \cdots,
\alpha_n) =  {\cal G}_\psi(\alpha_1, \cdots,\alpha_n) \\
 \forall (\alpha_1, \cdots, \alpha_n) \not \in \Lambda \; : \;  {\cal H}_{\varphi \wedge \psi}(\alpha_1, \cdots, \alpha_n) &:=&  {\cal H}_{\varphi \wedge \psi}(\alpha^*_1, \cdots,
\alpha^*_n) \;\; \mbox{
for some } \; (\alpha^*_1, \cdots,
\alpha^*_n)  \in \Lambda 
\end{eqnarray*}
The algorithm defines $(\alpha^*_1, \cdots,
\alpha^*_n)$ as the {\em ``nearest''} line of $(\alpha_1, \cdots, \alpha_n)$ in $\Lambda$.  This
 depends on the clause $\psi$.  Let $\psi =  [\neg] x_r
\vee [\neg] x_s \vee [\neg] x_t \;\; (1 \leq r < s < t \leq n)$, then   
\begin{eqnarray*}
(\alpha^*_1,\cdots,\alpha^*_n) :=
\left\{ \begin{array}{l}
(\alpha_1,
\cdots, \alpha_t+1, \cdots,\alpha_n) \; \; \mbox{ if }
(\alpha_1,
\cdots, \alpha_t+1, \cdots,\alpha_n)  \in \Lambda \;\; \\
\mbox{else }\\
\mbox{\hspace{0.6cm} } (\alpha_1,
\cdots, \alpha_{t-1}+1, \alpha_t, \cdots,\alpha_n) \; \; \mbox{ if }
(\alpha_1,
\cdots, \alpha_{t-1}+1, \alpha_t, \cdots,\alpha_n)  \in \Lambda \;\;\\
\mbox{\hspace{0.6cm} else } \\
\mbox{\hspace{1.3cm} }\cdots
\end{array}
\right.
\end{eqnarray*}
\end{theorem}
\mbox{}\\
 For instance, if  
\begin{flalign}
 [{\cal S}_\varphi]& =  
  \left[ {\tiny \begin{array}{ccccc}
& \alpha_1 & \alpha_2 & \alpha_3  \\
\hline
{\cal F}_\varphi & (\; 0 & 0 & 0 \;) \\
{\cal F}_\varphi & (\; 0 & 0 & 1 \;) \\
{\cal F}_\varphi & (\; 0 & 1 & 0 \;) \\
{\cal F}_\varphi & (\; 0 & 1 & 1 \;) \\
{\cal F}_\varphi & (\; 1 & 0 & 0 \;) \\
{\cal F}_\varphi & (\; 1 & 0 & 1 \;) \\
{\cal F}_\varphi & (\; 1 & 1 & 0 \;) \\
{\cal F}_\varphi & (\; 1 & 1 & 1 \;) \\
\end{array}
}
\right]
  =  \left[ {\tiny \begin{array}{ccc}
x_1 & x_2 & x_3 \\
\hline
\fcolorbox{gris}{gris}{0} & \fcolorbox{gris}{gris}{0} & \fcolorbox{black}{gris}{1} \\
\fcolorbox{gris}{gris}{0} & \fcolorbox{gris}{gris}{0} & \fcolorbox{gris}{gris}{1} \\
0 & 1 & \fbox{1} \\
0 & 1 & 1 \\
1 & 0 & 0 \\
1 & 0 & 1 \\
1 & 1 & \fbox{1} \\
1 & 1 & 1 \\
\end{array}
}
\right]  =  \left[ {\tiny \begin{array}{cccc}
x_1 & x_2 & x_3 \\
\hline
\fcolorbox{gris}{gris}{0} & \fcolorbox{gris}{gris}{0} & \fcolorbox{gris}{gris}{1} \\
0 & 1 & 1 \\
1 & 0 & 0 \\
1 & 0 & 1 \\ 
1 & 1 & 1 \\
\end{array}
}
\right] \nonumber \\[12pt]
\mbox{and } {\cal G}_{\psi} &\equiv \left[ \begin{array}{ccc}
g_1(\alpha_1) &
g_2(\alpha_1, \alpha_2) &
g_3(\alpha_1, \alpha_2, \alpha_3)
\end{array}
\right] \equiv \left[ \begin{array}{ccc}
\alpha_1 &
\alpha_2 &
\alpha_1 \alpha_3 + \alpha_2 \alpha_3 +\alpha_1 \alpha_2 \alpha_3
\end{array}
\right] \label{example}
\end{flalign}
\begin{flalign}
 \Rightarrow \; [{\cal S}_{\psi}]& =
\left[ {\tiny \begin{array}{ccccc}
& \alpha_1 & \alpha_2 & \alpha_3  \\
\hline
{\cal G}_\psi & (\; 0 & 0 & 0 \;) \\
{\cal G}_\psi & (\; 0 & 0 & 1 \;) \\
{\cal G}_\psi & (\; 0 & 1 & 0 \;) \\
{\cal G}_\psi & (\; 0 & 1 & 1 \;) \\
{\cal G}_\psi & (\; 1 & 0 & 0 \;) \\
{\cal G}_\psi & (\; 1 & 0 & 1 \;) \\
{\cal G}_\psi & (\; 1 & 1 & 0 \;) \\
{\cal G}_\psi & (\; 1 & 1 & 1 \;) \\
\end{array}
}
\right]
  =  \left[ {\tiny \begin{array}{ccc}
x_1 & x_2 & x_3 \\
\hline
0 & 0 & 0 \\
0 & 0 & \fbox{0} \\
0 & 1 & 0 \\
0 & 1 & 1 \\
1 & 0 & 0 \\
1 & 0 & 1 \\
1 & 1 & 0 \\
1 & 1 & 1 \\
\end{array}
}
\right]  =  \left[ {\tiny \begin{array}{cccc}
x_1 & x_2 & x_3 \\
\hline
0 & 0 & 0 \\
0 & 1 & 0 \\
0 & 1 & 1 \\ 
1 & 0 & 0 \\ 
1 & 0 & 1 \\ 
1 & 1 & 0 \\
1 & 1 & 1 \\ 
\end{array}
}
\right] \nonumber \\
& \mbox{\em Remark : The forbidden values $(\alpha_r^*,\alpha_s^*,\alpha_t^*)$ for $\psi$ are $\fcolorbox{gris}{gris}{\em (0\;,\;0\;,\;1)}$.} \nonumber
\end{flalign} 
\begin{flalign}
\mbox{then} \;\; [{\cal S}_{\varphi \wedge \psi}]& =  
 \left[ {\tiny \begin{array}{ccccc}
& \alpha_1 & \alpha_2 & \alpha_3  \\
\hline
{\cal H}_{\varphi \wedge \psi} & (\; 0 & 0 & 0 \;) \\
{\cal H}_{\varphi \wedge \psi} & (\; 0 & 0 & 1 \;) \\
{\cal H}_{\varphi \wedge \psi} & (\; 0 & 1 & 0 \;) \\
{\cal H}_{\varphi \wedge \psi} & (\; 0 & 1 & 1 \;) \\
{\cal H}_{\varphi \wedge \psi} & (\; 1 & 0 & 0 \;) \\
{\cal H}_{\varphi \wedge \psi} & (\; 1 & 0 & 1 \;) \\
{\cal H}_{\varphi \wedge \psi} & (\; 1 & 1 & 0 \;) \\
{\cal H}_{\varphi \wedge \psi} & (\; 1 & 1 & 1 \;) \\
\end{array}
}
\right] 
 =  \left[ {\tiny \begin{array}{cccc}
&x_1 & x_2 & x_3 \\
\hline
{\cal F}_{\varphi} (0,1,0) = & 0 & 1 & 1  \\   
{\cal F}_{\varphi} (0,1,0) = & 0 & 1 & 1  \\   
{\cal F}_{\varphi} (0,1,0) = & 0 & 1 & 1  \\   
{\cal F}_{\varphi} (0,1,1) = & 0 & 1 & 1  \\   
{\cal F}_{\varphi} (1,0,0) = & 1 & 0 & 0  \\   
{\cal F}_{\varphi} (1,0,1) = & 1 & 0 & 1  \\   
{\cal F}_{\varphi} (1,1,0) = & 1 & 1 & 1  \\   
{\cal F}_{\varphi} (1,1,1) = & 1 & 1 & 1  \\   
\end{array}
}
\right] 
{\tiny \begin{array}{c}
 \\
 \\
\mbox{\tiny ($1^{st}$ nearest line)}\\
\mbox{\tiny ($2^{nd}$ nearest line)}\\
\\
\\
\\
\\   
\\
\end{array}
} \nonumber \\
& =  \left[ {\tiny \begin{array}{cccc}
x_1 & x_2 & x_3 \\
\hline
0 & 1 & 1 \\
1 & 0 & 0 \\
1 & 0 & 1 \\
1 & 1 & 1 \\
\end{array}
}
\right] \nonumber
\end{flalign}
\mbox{}\\
\begin{proof} \\
The merging of $[{\cal S}_{\varphi}]$ and $[{\cal S}_{\psi}]$
should correspond to the intersection of the sets of solutions
$[{\cal S}_{\varphi}] \cap [{\cal S}_{\psi}]$.  In terms of ${\cal
  S}_\varphi$-matrices, this means that
only the lines common to  $[{\cal S}_{\varphi}]$ and $[{\cal
  S}_{\psi}]$ should be retained in $[{\cal S}_{\varphi \wedge
  \psi}]$.  As these lines are, by definition of the descriptor
function, the elements of
$Im \; {\cal F}_{\varphi}$ and $Im \; {\cal G}_{\psi}$, we have to
get $Im \; {\cal
H}_{\varphi \wedge \psi} = Im \; {\cal F}_{\varphi} \; \cap \; Im
\; {\cal G}_{\psi}$.  
\mbox{} \\
Let $\psi = [\neg] x_r
\vee [\neg] x_s \vee [\neg] x_t$ where $1 \leq r < s < t \leq
n$ be a 3-CNF clause and $(\alpha_r^*,\alpha_s^*,\alpha_t^*)$ be the
unique triplet of non satisfying
values for $x_r,\; x_s$ and $x_t$.   The clause $\psi$ puts a sole
constraint over $\varphi$, in the sense that we have to discard
 the lines of
$[{\cal S}_\varphi]$ including the forbidden values $(\alpha_r^*,\alpha_s^*,\alpha_t^*)$.  Three situations can occur. \\
\mbox{}\\
\noindent $\bullet$ {\bf Situation A} : All lines of  $[{\cal S}_{\varphi}]$  can be
kept to build $[{\cal S}_{\varphi \wedge \psi}]$, as
none of them includes the forbidden values  
$(\alpha_r^*,\alpha_s^*,\alpha_t^*)$. So, $\psi$ does not introduce any new constraint with respect to $\varphi$ and 
$\;[{\cal S}_{\varphi \wedge \psi}] := [{\cal S}_\varphi]\;$ or $\;{\cal H}_{\varphi \wedge \psi} := {\cal F}_\varphi \;\;$.\\
\mbox{}\\
Let us look
at situations where some lines of $[{\cal S}_{\varphi}]$ includes the forbidden values 
$(\alpha_r^*,\alpha_s^*,\alpha_t^*)$. We define the {\em ``nearest''}
line in $\{0,1\}^n$ as the line having all the same $\alpha_i$, except
for $\alpha_t$ where we have $\alpha_t +1$. Let us look at a line containing the forbidden values $(\alpha_r^*,\alpha_s^*,\alpha_t^*)$. Two situations can occur :  \\
\mbox{}\\
\noindent $\bullet$ {\bf Situation B} : The image by ${\cal F}_\varphi$ of its nearest line 
{\em does not include} the forbidden va\-lues
$(\alpha_r^*,\alpha_s^*,\alpha_t^*)$.  All the $\alpha_i \; (i \neq
t)$ are the same as in the original line, except $\alpha_t$ that becomes $\alpha_t+1$. The algorithm
give us the solution : ${\cal H}_{\varphi \wedge \psi}(\alpha_1, \cdots,\alpha_t,\cdots, \alpha_n) := {\cal F}_{\psi} (\alpha_1,\cdots,\alpha_t+1,\cdots,\alpha_n)$  \\[12pt]
\noindent $\bullet$ {\bf Situation C} : The image by ${\cal F}_\varphi$ of the nearest line 
{\em does include} the forbidden values
$(\alpha_r^*,\alpha_s^*,\alpha_t^*)$.   This corresponds to the above example where the {\em first nearest line} of the line in grey [${\cal F}_\varphi (0,0,0) = (0,0,1)$] also includes the forbidden values [${\cal F}_\varphi (0,0,1) = (0,0,1)$].  We need to find a {\em second nearest line}, defined as the line having the same $\alpha_i$ as the original line, except for
$\alpha_{t-1}$ where it is the opposite. In our example, this is the third line, where ${\cal F}_\varphi (0,1,0) = (0,1,0)$, a line corresponding to a solution.  
{\bf \em But} it might be necessary to
look at successive nearest lines before finding a line without the forbidden values and thus a solution
for ${\cal H}_{\varphi \wedge
  \psi}$. {\em Otherwise, the algorithm stops and the 3-CNF formula $\varphi \wedge \psi$  is
without solution.} \\[12pt]
\noindent  The three situations can be summarized.  The solution for ${\cal H}_{\varphi \wedge
  \psi}$ will correspond to the following algorithm :
\begin{eqnarray*}
{\cal H}_{\varphi \wedge \psi}(\cdot) &:=& \left\{ \begin{array}{l}
{\cal F}_{\varphi}(\alpha_1,\cdots,\alpha_t, \cdots, \alpha_n) \;\;\; \mbox{\hspace{4.5cm} (situation A)}\\
 \mbox{\hspace{0.5cm}  when } (f_r(\cdot),f_s(\cdot),f_t(\cdot)) \neq
(\alpha_r^*,\alpha_s^*,\alpha_t^*) \\
 \mbox{} \stackrel{\mbox{\tiny \em see (\ref{example})}}{\Leftrightarrow}\; 
g_t(f_r(\alpha_1,\cdots,\alpha_r),f_s(\alpha_1,\cdots,\alpha_s),f_t(\alpha_1,\cdots,\alpha_t)) = f_t(\alpha_1,\cdots,\alpha_t) \\
\\
{\cal F}_{\varphi}(\alpha_1,\cdots,\alpha_t+1, \cdots, \alpha_n) \;\;\; \mbox{\hspace{3.9cm} (situation B)}\\
 \mbox{\hspace{0.5cm}  when } (f_r(\cdot),f_s(\cdot),f_t(\cdots,\alpha_t)) =
(\alpha_r^*,\alpha_s^*,\alpha_t^*) \\
 \mbox{\hspace{0.5cm}  and } (f_r(\cdot),f_s(\cdot),f_t(\cdots,\alpha_t+1)) \neq
(\alpha_r^*,\alpha_s^*,\alpha_t^*) \\
\mbox{} \stackrel{\mbox{\tiny \em see (\ref{example})}}{\Leftrightarrow}\;  g_t(f_r(\alpha_1,\cdots,\alpha_r),f_s(\alpha_1,\cdots,\alpha_s),f_t(\alpha_1,\cdots,\alpha_t)) \neq f_t(\alpha_1,\cdots,\alpha_t) \\
 \mbox{\hspace{1cm} and } g_t(f_r(\cdot),f_s(\cdot),f_t(\alpha_1,\cdots,\alpha_t+1)) = f_t(\alpha_1,\cdots,\alpha_t+1) \\
\\
{\cal F}_{\varphi}(\alpha_1,\cdots,\alpha_j+1,\cdots,\alpha_t, \cdots,\alpha_n) \;\;\; \mbox{\hspace{2.9cm} (situation C)}\\
\mbox{\hspace{0.5cm}  when } (f_r(\cdot),f_s(\cdot),f_t(\cdots,\alpha_t)) =
(\alpha_r^*,\alpha_s^*,\alpha_t^*) \\
 \mbox{\hspace{0.5cm}  and } (f_r(\cdot),f_s(\cdot),f_t(\cdots,\alpha_t+1)) =
(\alpha_r^*,\alpha_s^*,\alpha_t^*) \\
 \mbox{\hspace{0.5cm}  and } (f_r(\cdot),f_s(\cdot),f_t(\cdots,\alpha_j+1,\cdots,\alpha_t)) \neq
(\alpha_r^*,\alpha_s^*,\alpha_t^*) \\
 \mbox{\hspace{1cm} for some $\alpha_{j} (j<t)$ 
   so that situation A or B arises. We take the} \\
\mbox{\hspace{1cm} unique highest such $\alpha_j$. If not existing, $\varphi \wedge \psi$ has  no solution.}
\end{array}
\right.\\
\end{eqnarray*}
\end{proof}
\mbox{}\\
\begin{theorem}{\bf General computation theorem for 3-CNF formula} \label{computation_algo} \\[12pt] 
{\bf The descriptor function for the conjunction of a
  3-CNF formulae $\varphi$ and a 3-CNF clause $\psi =  [\neg] x_r
\vee [\neg] x_s \vee [\neg] x_t \;\;(1 \leq r < s < t \leq n)$ can be
numerically computed as the merging of their
  descriptor functions : $[ {\cal H}_{\varphi \wedge \psi} ] = [
  {\cal F}_{\varphi} ] \wedge [ {\cal G}_{\psi} ]$.}
\end{theorem}
%
\begin{eqnarray*} \mbox{ More precisely, if } [{\cal F}_\varphi] \equiv \raisebox{5ex}{t} \left[ \begin{array}{c}
f_1(\alpha_1)\\
\vdots\\
f_n(\alpha_1,\cdots,\alpha_n)
\end{array}
\right] \mbox{ and } [{\cal G}_{\psi}] \stackrel{\mbox{\tiny \em see (\ref{def_h})}}{\equiv} \raisebox{9ex}{t} \left[ \begin{array}{c}
g_1(\alpha_1) =\alpha_1 \\
\vdots\\
g_t(\alpha_r,\alpha_s,\alpha_t) \\
\vdots \\
g_n(\alpha_1,\cdots,\alpha_n) = \alpha_n
\end{array}
\right] \\
\end{eqnarray*}
\begin{eqnarray*} \mbox{ Then } [ {\cal H}_{\varphi \wedge \psi} ]  &\equiv& \raisebox{5ex}{t} \left[ \begin{array}{c} 
h_1(\alpha_1)\\
\vdots\\ 
h_n(\alpha_1,\cdots,\alpha_n)
\end{array}
\right] \\
& & \\
&\stackrel{\tiny notation}{=}&
\raisebox{5ex}{t} \left[ \begin{array}{c} 
f_1(\alpha_1) \wedge g_1(\alpha_1)\\
\vdots\\ 
f_n(\alpha_1,\cdots,\alpha_n) \wedge g_n(\alpha_1,\cdots,\alpha_n)
\end{array}
\right] \mbox{\hspace{1cm}}
\end{eqnarray*}
\mbox{}\\
\noindent where, in a loop for $l$ going from $n$ to $1$ : 
\begin{eqnarray}
& & \hspace{-3.8cm} \mbox{Let } \beta_i \stackrel{\mbox{\tiny \em notation}}{\equiv} f_i(\alpha_1,\cdots, \alpha_i) \;,  \hspace{6.5cm} \mbox{[Situations A and B]}\nonumber \\
h_l(\alpha_1,\cdots, \alpha_l) 
:= & (\alpha_l +1) &\; \cdot \;  \{ \; [f_l(\alpha_1,\cdots, \alpha_{l-1},0) +  
 g_l(\beta_1,\cdots, \beta_{l-1},0)] \label{merge} \\
& &  \; \; \;  \; \;\cdot \; [ f_l(\alpha_1,\cdots, \alpha_{l-1},1) \cdot  
 g_l(\beta_1,\cdots, \beta_{l-1},1)]  \nonumber \\
& &  \; \; + \; [f_l(\alpha_1,\cdots, \alpha_{l-1},0) \cdot  
 g_l(\beta_1,\cdots, \beta_{l-1},0)] \; \} \nonumber \\
&  \; +\;\;  \; \alpha_l &\; \cdot \;
   \{ \; [f_l(\alpha_1,\cdots, \alpha_{l-1},1) +  
 g_l(\beta_1,\cdots, \beta_{l-1},1)]  \nonumber \\
& &  \; \;\;  \; \;\cdot \; [f_l(\alpha_1,\cdots, \alpha_{l-1},0) +  
 g_l(\beta_1,\cdots, \beta_{l-1},0) ]  \nonumber \\
& &  \; \; \; + \; [f_l(\alpha_1,\cdots, \alpha_{l-1},1) +  
 g_l(\beta_1,\cdots, \beta_{l-1},1)]   \nonumber \\
& &  \;  \; \; \;  \; \cdot \; [f_l(\alpha_1,\cdots, \alpha_{l-1},0) \cdot  
 g_l(\beta_1,\cdots, \beta_{l-1},0)]   \nonumber \\
& & \;  \; \;+   \; [f_l(\alpha_1,\cdots, \alpha_{l-1},1) \cdot  
 g_l(\beta_1,\cdots, \beta_{l-1},1) ] \} \nonumber
\end{eqnarray}
\begin{eqnarray}
\mbox{Moreover if there exists} & \mbox{ a } & j < l \mbox{  ,
   related to the highest $\alpha_j$ (j is thus unique), such
   that :} \; \; \nonumber \\
  g^*_j(\alpha_1,\cdots, \alpha_j)& \equiv  & 
  [f_l(\alpha_1,\cdots, \alpha_{l-1},0) +  
 g_l(\beta_1,\cdots, \beta_{l-1},0)]  \; \cdot \; \hspace{1.5cm} \mbox{[Situation C]}\nonumber \\
& & [ f_l(\alpha_1,\cdots, \alpha_{l-1},1) +  
 g_l(\beta_1,\cdots, \beta_{l-1},1)] \;\; \mbox{\bf  = 1}  \label{recursive} \\
& & \mbox{[(\ref{recursive}) corresponds to situations where ${\cal F}_\varphi(\alpha_1,\cdots,\alpha_{l-1}+1,\alpha_l)$} \nonumber \\
& & \mbox{includes the forbidden values.  We look at the {\em nearest} $l$-uple } \nonumber \\
& & \mbox{not including the forbidden values,  where we put $\alpha_j:=\alpha_j+1$]} \nonumber \\
& & \hspace{-1.4cm} \Rightarrow \;\; \mbox{\bf then replace } f_j(\alpha_1, \cdots, \alpha_j) \; \mbox{ in } \; [{\cal F}_\varphi] \nonumber \\
& & \hspace{-1cm} \; \; \mbox{by this new merging } 
\; \; f_j(\alpha_1, \cdots, \alpha_j) \;\; \wedge
  \;\{\; g^*_j(\alpha_1, \cdots, \alpha_j) + \alpha_j \; \} \label{g*_j} \\
& & \hspace{-1cm} \; \; \mbox{computed by using a recursive call to definition
  (\ref{merge})}. \nonumber 
\end{eqnarray}
\noindent Recursivity will end as soon as there is no longer such
$g^*_j(\alpha_1,\cdots, \alpha_j) = 1$. When $
g^*_j(\alpha_1,\cdots, \alpha_j)$ is no longer a function of
$\alpha_i$, that means that the 3-CNF-SAT $\varphi \wedge \psi$ has no solution.
\begin{flalign}
&\mbox{At the end of the loop, we replace 
$\alpha_i$ by $h_i(\alpha_1, \cdots,
\alpha_i)$ in $h_j(\cdot)$ where $1 \leq i < j \leq
n$.} & \label{updating} 
\end{flalign}
\mbox{} \\
\begin{proof} 
\noindent This computational formula  gives the same
answer for ${\cal H}_{\varphi \wedge \psi}$ as the algorithm in {\em Theorem \ref{algorithm}}. {\em Remember that $(\alpha_r^*,\alpha_s^*,\alpha_t^*)$ are the forbidden values for $\psi$.}
\mbox{}\\
\noindent More formally, four situations in equation
(\ref{merge}) should be considered for the index $t$ : 
\begin{eqnarray*}
\bullet \; f_t(\alpha_1,\cdots,0)&=& g_t(\beta_r,\beta_s,0) \;\;
\mbox{ and } \; f_t(\alpha_1,\cdots,1) = g_t(\beta_r,\beta_s,1) \hspace{1.5cm} \mbox{[as in situation A]}\\ 
&\Downarrow & \\
h_t(\alpha_1,\cdots,\alpha_{t}) & := & 
(\alpha_t +1) \cdot \;  \{ \; [f_t(\cdot,0) +  g_t(\cdot,0)]  \; \cdot \; [ f_t(\cdot,1) \cdot  g_t(\cdot,1)] \; + \; [f_t(\cdot,0) \cdot  
 g_t(\cdot,0)] \; \}  \; + \\
& & \; \alpha_t \cdot   \{ \; [f_t(\cdot,1) +  
 g_t(\cdot,1)]  \; \cdot \; [f_t(\cdot,0) +  
 g_t(\cdot,0)  ]   \; + \\
& & \; \;\;\;\;\;\;\;\;\;\;\; [f_t(\cdot,1) +  
 g_t(\cdot,1)]  \; \cdot \; [f_t(\cdot,0) \cdot  
 g_t(\cdot,0)]  \; +   \; [f_t(\cdot,1) \cdot  
 g_t(\cdot,1) ] \;  \} \\
& := & (\alpha_t +1) \cdot  f_t(\cdot,0) +
 \alpha_t \cdot  f_t(\cdot,1)   \;\;\; \mbox{[as $f_t()+ g_t()=0$ ;
   $f_t() \cdot g_t() = f_t^2()=f_t()$]} \\
& := & f_t(\alpha_1,\cdots,\alpha_{t}) \\
& & \hspace{-2.5cm} \mbox{[$h_t()$ is thus the merging of the cells $f_t()$ and
  $g_t()$ which are equivalent.]} \\[24pt]
\end{eqnarray*}
\begin{eqnarray*}
\bullet \; f_t(\alpha_1,\cdots,0) \;\; &=& g_t(\beta_r,\beta_s,0)  \;\; \mbox{ but } \;\;
f_t(\alpha_1,\cdots,1) \neq g_t(\beta_r,\beta_s,1)  \hspace{0.9cm} \mbox{[as in situation B]}\\ 
&\Downarrow & \\
h_t(\alpha_1,\cdots,\alpha_t) & := & (\alpha_t +1) \cdot  f_t(\cdot,0) +
 \alpha_t \cdot  f_t(\cdot,0)   \;\;\; \mbox{[as $f_t(\cdot,1)+ g_t(\cdot,1)=1$ and} \\
& & \hspace{5cm} \mbox{   $f_t(\cdot,1) \cdot g_t(\cdot,1) = 0$]} \\
& := & f_t(\alpha_1,\cdots,0)\\
& & \hspace{-2.8cm} \mbox{[$h_t()$ sends $\alpha_t$ to a value where the cells $f_t()$ and
  $g_t()$ are the same  in
  $[{\cal F}_{\varphi}]$ and $[{\cal G}_{\psi}]$]} \\[24pt]
\bullet \; f_t(\alpha_1,\cdots,1) \;\; &=& g_t(\beta_r,\beta_s,1)  \;\; \mbox{ but } \;\;
f_t(\alpha_1,\cdots,0) \neq g_t(\beta_r,\beta_s,0) \hspace{0.9cm} \mbox{[as in situation B]} \\ 
&\Downarrow & \\
h_t(\alpha_1,\cdots,\alpha_t) & := & (\alpha_t +1) \cdot  f_t(\cdot,1) +
 \alpha_t \cdot  f_t(\cdot,1)   \;\;\; \mbox{[as $f_t(\cdot,0)+
   g_t(\cdot,0)=1$ and } \\
& & \hspace{5cm} \mbox{   $f_t(\cdot,0) \cdot g_t(\cdot,0) = 0$]} \\
& := & f_t(\alpha_1,\cdots,1) \\
& & \hspace{-2.8cm}  \mbox{[$h_t()$ sends $\alpha_t$ to a value where the cells $f_t()$ and
  $g_t()$ are the same  in 
  $[{\cal F}_{\varphi}]$ and $[{\cal G}_{\psi}]$]} \\[24pt]
%
\bullet \; f_t(\alpha_1,\cdots,1) \;\; & \neq &g_t(\beta_r,\beta_s,1)  \;\; \mbox{ and } \;\;
f_t(\alpha_1,\cdots,0) \neq g_t(\beta_r,\beta_s,0) \hspace{0.9cm} \mbox{[as in situation C]} \\ 
&\Downarrow & \;\;
\mbox{[Impossibility to find a common cell between $\; f_t() \;$ and $\; g_t()
    \;$]} \\
&    \Downarrow & \;\;
\mbox{[No constraint for $h_t(\cdot,\alpha_t)$ but  a induced constraint} \\
& & \;\; \mbox{ over
  some {\em predecessor} $\alpha_j \; (j<t)$]} \\
h_t(\alpha_1,\cdots,\alpha_t) & := &  \alpha_t    \;\;\; \mbox{[as $f_t()+ g_t()=1$ and
   $f_t() \cdot g_t() = 0$]} \\
&\mbox{\bf but} & \;\; 
[f_t(\cdot,0) +  
 g_t(\cdot,0)]  \; \cdot \;   [f_t(\cdot,1) +  g_t(\cdot,1)] = \mbox{function}(\alpha_1,\cdots, \alpha_j) =1 \\
\mbox{\bf and} \;\;\;
  &g_j(\alpha_1,&\cdots, \alpha_j) \leftarrow   
  [f_t(\cdot,0) +  
 g_t(\cdot,0)]  \; \cdot \;   [f_t(\cdot,1) +  
 g_t(\cdot,1)] +  g_j(\alpha_1,\cdots, \alpha_j) \\
& &  \hspace{-1.5cm} \mbox{[New additional constraint over
  $g_j(\cdot,\alpha_j) \; ( j < t)$ as $g_j(\cdot, \alpha_j) \leftarrow   g_j(\cdot, \alpha_j) +1.$} \\
& & \hspace{-1.5cm} \mbox{{\bf A descending order with respect of $t$} for the
  computations ensures us  } \\
& & \hspace{-1.5cm} \mbox{that the new additional constraint over
  $g_j(\cdot,\alpha_j)$ has no repercussion over the } \\
& & \hspace{-1.5cm} \mbox{already-computed function  $h_l(\cdot,\alpha_l) \;[l\ge t]$, as $\alpha_l$ is not
  involved in $g_j(\cdot,\alpha_j)$.]}
\end{eqnarray*}
\end{proof}
\noindent {\bf Note :} The code for this merging operation is available at
{\em  https://github.com/3cnf/} in the {\em descriptor-solver} directory. \\
\subsection {Examples of computation} 
$1) \;\;${\it Example of a simple merging of two clauses} \\
\noindent Let $\;  \varphi  = (x_1 \vee x_2 \vee \neg x_3) \;\; \mbox{
  and } \;\; \psi  =  (\neg x_2 \vee  x_3
\vee \neg x_4) $ \hspace{12pt}  (see \ref{example2}) \\
\begin{flalign*}
 [{\cal F}_\varphi]& = 
\raisebox{7ex}{t}  
\left[ \begin{array}{c}
\alpha_1 \\
 \alpha_2 \\
 (\alpha_1+1)(\alpha_2+1)\alpha_3+\alpha_3 \\
\alpha_4
\end{array}
\right] = \raisebox{7ex}{t} \left[ \begin{array}{c}
\beta_1 \\
\beta_2 \\
\beta_3 \\
\beta_4
\end{array}
\right] 
\hspace{12pt} [{\cal G}_{\psi}] =  \raisebox{7ex}{t} \left[ \begin{array}{c}
g_1(\cdot) \\ 
g_2(\cdot) \\ 
g_3(\cdot) \\ 
g_4(\cdot)  
\end{array}
\right] =  \raisebox{7ex}{t} \left[ \begin{array}{c}
\alpha_1 \\ 
\alpha_2 \\ 
\alpha_3 \\
\alpha_2 (\alpha_3+1) \alpha_4 +\alpha_4
\end{array}
\right]  \\
\\
& \Rightarrow [{\cal H}_{\varphi \wedge \psi}] =  \raisebox{7ex}{t} \left[ \begin{array}{c}
h_1(\cdot) \\ 
h_2(\cdot) \\ 
h_3(\cdot) \\ 
h_4(\cdot)  
\end{array}
\right] = \raisebox{7ex}{t} \left[ \begin{array}{c}
\alpha_1 \\ 
\alpha_2 \\ 
 (\alpha_1+1)(\alpha_2+1)\alpha_3+\alpha_3 \\
\alpha_2 (\alpha_3+1) \alpha_4 +\alpha_4
\end{array}
\right] 
\end{flalign*}
\mbox{}\\ 
\mbox{}\\ 
\noindent {\bf Computations for} $h_t(\cdot)$ [Descending order for $t$] :  \\
$\bullet \; \; \underline{t=4} \;\;: \;\;
g_4(\beta_i)= \alpha_2 \cdot \{ [
  (\alpha_1+1)(\alpha_2+1)\alpha_3+\alpha_3]+1\} \cdot \alpha_4
+\alpha_4 = \alpha_2 \alpha_3 \alpha_4 + \alpha_2 \alpha_4 +\alpha_4$\\
$h_4(\cdot)  \stackrel{(\ref{merge})}{=}
(\alpha_4+1) \{[0+0] \cdot [1 \cdot (\alpha_2 \alpha_3 + \alpha_2 + 1)] +
[0 \cdot 0] \} 
+ \alpha_4 \{[1 + \alpha_2 \alpha_3 + \alpha_2 + 1] \cdot [0+0] +
[1 + \alpha_2 \alpha_3 + \alpha_2 + 1] \cdot [0 \cdot 0]
+ [1 \cdot (\alpha_2 \alpha_3 + \alpha_2 + 1)]\} = \alpha_2 \alpha_3
\alpha_4 + \alpha_2 \alpha_4 +\alpha_4 = g_4(\cdot) $ \\
Indeed, the merging of $f_4(\cdot)$ and $g_4(\cdot)$ should give
$g_4(\cdot)$ as $f_4(\cdot)$ puts no constraint over $x_4$.\\[12pt]
$\bullet \;\;  \underline{t=3} \;\;: \;\; h_3(\alpha_1,\alpha_2,\alpha_3) =
f_3(\alpha_1,\alpha_2,\alpha_3) $ as $g_3(\beta_i)=\beta_3=f_3(\cdot)$\\
Here, $g_3(\cdot)$ puts no constraint over $x_3$.  Idem for
$h_2(\cdot)$ and $h_1(\cdot)$.\\

$2) \;\; $ {\it Example of a uniformly distributed 3-CNF-SAT problem :} 
\begin{eqnarray}
\mbox{Let }\;\; \varphi := \bigwedge_{i=1}^{8} \psi_i = \bigwedge \left\{ \begin{array}{l}
x_1 \vee \neg x_2 \vee \neg x_3 \\
 x_1 \vee  x_2 \vee \neg x_3 \\
\neg x_1 \vee \neg x_2 \vee \neg x_3 \\
\neg x_1 \vee  x_2 \vee \neg x_3 \\
x_1 \vee \neg x_2 \vee  x_3 \\
x_1 \vee x_2 \vee x_3 \\
\neg x_1 \vee \neg x_2 \vee x_3 \\
\neg x_1 \vee x_2 \vee x_3 
\end{array}
\right. \label{8-clauses}
\end{eqnarray}
We
have here  :
\begin{eqnarray*}
\begin{array}{llllccc}
\mbox{Step \hspace{1cm} } &h_1(\cdot)&h_2(\cdot)&h_3(\cdot)  & \# \{
{\cal S}_\varphi \} \\
\hline \hline
\varphi = \psi_1& \alpha_1 & \alpha_2 &
(\alpha_1 + 1) \alpha_2 \alpha_3 +\alpha_3  &7 \\
\varphi = \psi_1 \wedge \psi_2 &\alpha_1 &  \alpha_2 &
\alpha_1 \alpha_3 & 6\\
\varphi = \wedge_{i=1}^{3} \psi_i & \alpha_1 & \alpha_2
& \alpha_1 \alpha_2 \alpha_3 + \alpha_1 \alpha_3 & 5\\ 
\varphi = \wedge_{i=1}^{4} \psi_i  & \alpha_1 & \alpha_2 &
 0 &  4 \\
\varphi = \wedge_{i=1}^{5} \psi_i  & \alpha_1 & \alpha_1 \alpha_2 & 0
& 3 \\
\varphi = \wedge_{i=1}^{6} \psi_i  & 1 &  \alpha_2 & 0 &  2 \\
\varphi = \wedge_{i=1}^{7} \psi_i  & 1 & 0 & 0 &  1 \\
\varphi = \wedge_{i=1}^{8} \psi_i  & \nexists & \nexists & \nexists  &
0 \\
\end{array} 
\\[12pt]
\end{eqnarray*}
\begin{property}Let $A, B \subseteq \{\alpha_1,\cdots,\alpha_n\}.$  If $f_t(\cdot) \in \combi(A)$ and $g_t(\cdot) \in \combi(B)$, then
\begin{eqnarray}
& &f_t(\cdot) \wedge g_t(\cdot) \in \combi(A \cup
B) \label{combi_wedge} \\
&\mbox{and}& g_j^*(\cdot) \in \combi([A \cup B] \setminus
\{\alpha_{j+1}, \cdots, \alpha_n\}) \label{combi_g} \\ \nonumber
\end{eqnarray}
\end{property}
\begin{proof}
This is straightforward from the definition of $f_t(\cdot) \wedge
g_t(\cdot)$ in (\ref{merge}) and $g_{j}^*(\cdot)$ in (\ref{recursive}).
\end{proof}
\section{Complexity analysis for computing ${\cal H}_\varphi$, the
    descriptor function}
\subsection{Some definitions} 
\begin{definition}{\bf Sorted clauses (to ensure the descendant order of $t$)} \label{sorted} \\[12pt]
\noindent Let  $\; \varphi = \bigwedge_{k=1}^m \psi_k \;$ be a
  3-CNF formula.  We suppose, without any loss of generality, that
 these $m$ 3-CNF clauses are {\bf sorted}, in the following way :
\begin{eqnarray}
\varphi &=& \bigwedge_{k=1}^m \psi_k \nonumber \\
\mbox{ where  } &\;& \psi_k = [\neg] x_{r_k} \vee [\neg] x_{s_k} \vee [\neg] x_{t_k}
\;\; (1 \leq r_k < s_k < t_k \leq n) \nonumber \\
\mbox{ and } & \;&
\left\{ \begin{array}{rr}
 & t_k < t_{k'} \\
\mbox{or} \; \;\; t_k = t_{k'} &\mbox{ where } \psi_k = [\neg] x_{r_k} \vee [\neg] x_{s_k}
          \vee \neg \; x_{t_k} \\
&\mbox{ while } \psi_{k'} = [\neg]
          x_{r_{k'}} \vee [\neg] x_{s_{k'}} \vee  x_{t_{k'}} 
\end{array} \right\}   \; \; \Rightarrow \;\;  k < k' 
 \label{sorted-clauses}
\end{eqnarray}
\end{definition}
\mbox{}\\
\begin{definition}{\bf Set of predecessors $P(t)$} \\[12pt]
\noindent Let us compute $f_t(\cdot) \wedge g_t(\cdot)$.  We define
the indice $j$ of
the highest $\alpha_j$ found in the recursive call (\ref{recursive}) 
as the {\em predecessor} of $t$, $t$ being the successor of $j$.  In the same way, $x_j$ and $\alpha_j$ are
called the {\em predecessors} of $x_t$ and $\alpha_t$.  We denote them by :
\begin{eqnarray*}
j &\stackrel{def}{=}&  \pred (t) \; \mbox{ and } \; t
\stackrel{def}{=}  \succ(j)\\
x_j &\stackrel{def}{=}&  \pred (x_t) \; \mbox{ and } \;
\alpha_j \stackrel{def}{=}  \pred (\alpha_t) \\
j' &=& \pred^{\;2}(t) \; \mbox{ iff } \; j' =
\pred (\pred (t) ) \; \mbox{ and so on.} 
\end{eqnarray*}
\noindent Let us define $P(t)$ as {\em the set of the  predecessors} of $t$ :
\begin{eqnarray}
P(t) = \{ j < t \;|\; \exists k \; :\; j=\pred^{\;k}(t) \}  \label{P_t}
\end{eqnarray}
\end{definition}
\mbox{}\\
\begin{definition}{\bf Set of connected variables $V(x_t)$ and $W(x_t)$} \\[12pt]
\noindent We define for all
 $t$ in $\{3,\cdots,n\}$ : 
\begin{eqnarray} 
Cl(x_t)  &=& \{ \psi_k \; \mbox{ where the variable $x_t$ appears 
with the highest indice} \}  \label{Cl_t} \\
V(x_t)  &=& \{\alpha_i  \; (i \leq t)  \; | \; x_i \mbox{ appears in some } \psi_k  \in Cl(x_t) \} \label{V_t}
\\
V_{|\alpha_i}(x_t) &=& V(x_t) \setminus
\{\alpha_{i+1},\cdots\}  \label{V_i_t} \\
W^*(x_t) &=&  \left[ \bigcup_{u > t \; : \;  
  t \in P(u)} V(x_u) \right] \; \cup V(x_t)  \label{W*_t} \\
 W(x_t) &=& W^*(x_t) \setminus \{\alpha_{t+1}, \cdots,
\alpha_n \} \;\; [ \mbox{ with } W(x_n) \stackrel{def}{=} V(x_n) \; ] \label{W_t} \\
\nonumber
\end{eqnarray}
\noindent $V(x_t)$ is the set of $\alpha_i$ corresponding to the
variables connected to $x_t$ by at least one clause $\psi_k$ where $x_t$ is
  the highest indexed variable.  $W(x_t)$ is the union of $\alpha_i $
  connected to the successors of $x_t$, excluding $\alpha_i$
  with indice $i$ higher than $t$. 
\end{definition}
\mbox{}\\
\begin{definition}{\bf Sub-problem $\varphi^{(L)}$} \\[12pt]
\noindent We define the {\em sub-problem $\varphi^{(L)} $ associated to the subset $L \subseteq
  \{1, \cdots,m\} \;$ } by :
\begin{eqnarray}
\varphi^{(L)} = \bigwedge_{k\in L} \psi_k  \label{sub-problem}
\end{eqnarray}
 and $Cl^{(L)}(x_t)$, $V^{(L)}(x_t)$ and $W^{(L)}(x_t)$
 being the respective sets  $Cl(x_t), V(x_t)$ and $W(x_t)$
 for $\varphi^{(L)}$. In the same way, we define $P^{(L)}(t)$ as the
 set of the predecessors of $t$ for $\varphi^{(L)}$. See
 (\ref{P_t}). \\[12pt]
\noindent Finally, we define :
\[{\cal H}_{\varphi^{(L)}} = \left[ \begin{array}{c}
h_1^{(L)}(\cdot)\\
\vdots\\
h_n^{(L)}(\cdot) \\
\end{array}
\right] \]
\mbox{}\\
\end{definition}
\subsection{Complexity theorem for computing ${\cal H}_{\varphi}$} 
\begin{theorem} {\bf 3-CNF complexity theorem for ${\cal H}_{\varphi}$} \label{complexity_1}\\
\noindent The complexity of the descriptor
    approach for a 3-CNF problem $\varphi$ with $m$ clauses and $n$
    propositional variables is 
\[{\cal
  O}( m \; n^2 \; \max_{1 \leq t \leq n} \; \max_{1 \leq l \leq m} \len(h_t^{(\{l,\cdots,m\})}) \; ) \]
See (\ref{len-def}) for the definition of $\len(h_t)$.\\
\end{theorem}
\begin{proof}
\mbox{}\\
Let us compute the complexity of $f_t(\cdot) \wedge g_t(\cdot)$ in
(\ref{merge}) for the most general case.  First of all, one has to compute the four functions in square
brackets :
$[f_t(\cdot,0)+g_t(\cdot,0)]$, $[f_t(\cdot,1)+g_t(\cdot,1)]$, $[f_t(\cdot,0) \cdot
g_t(\cdot,0)]$ and $[f_t(\cdot,1) \cdot g_t(\cdot,1)$].  We have :
\begin{eqnarray*}
\len(f_t(\cdot,0)) \leq \len(f_t(\cdot,\alpha_t)) \;\; &\mbox{and}& \;\;
\len(f_t(\cdot,1)) \leq \len(f_t(\cdot,\alpha_t)) \; [\equiv \len(f_t)] \\
\len(g_t(\cdot,0)) \leq \len(g_t(\cdot,\alpha_t)) \;\; &\mbox{and}& \;\; 
\len(g_t(\cdot,1)) \leq \len(g_t(\cdot,\alpha_t)) \; [\equiv \len(g_t)] \\
\len(f_t + g_t)  \leq \len(f_t) +
\len(g_t) &\leq& 
\len(f_t) \cdot \len(g_t) \;\;\;\; \mbox{when} \;\;
\len(f_t) > 2 \mbox{ and } \len(g_t) > 2
\end{eqnarray*} 
The complexity for the four functions is then ${\cal O}(
\len(f_t) \cdot \len(g_t ))$. \\
\noindent The complexity for computing $h_t(\cdot)$ in (\ref{merge})
is :
\begin{eqnarray}
{\cal O}(&3 \; \cdot& [( \len(f_t) \cdot \len(g_t) )^2 + ( \len(f_t) \cdot
\len(g_t) ) ] + 2 \cdot [ 2 ( \len(f_t) \cdot \len(g_t) )^2 + ( \len(f_t) \cdot
\len(g_t) )] ) \nonumber \\
&=&  {\cal O}( 7 \cdot ( \len(f_t) \cdot \len(g_t) )^2 + 5 \cdot ( \len(f_t) \cdot
\len(g_t) ) ) \nonumber \\
&=& {\cal O}( [\len(f_t) \cdot \len(g_t)]^2 ) \;\;\; \mbox{ for large }
\len(f_t) \cdot \len(g_t)   \label{complexity}
\end{eqnarray}
\noindent Note : it needs three
runs over the formula in the brackets to do the product with
$(\alpha_t +1)$ :
one to compute the formula, one to multiply it by $\alpha_t$ and one to
add both results. Similarly, it takes two runs to compute the product with $\alpha_t$.\\[12pt]
\noindent Using the same argumentation, we have :
\begin{eqnarray}
\len(h_t) &=& {\cal O}( [\len(f_t) \cdot \len(g_t)]^2 ) \;\;\; \mbox{ for large }
\len(f_t) \cdot \len(g_t)   \label{len_h} \\
& &\hspace{-3cm} \mbox{ and  for the recursive call with } j < t  \;\;\; \mbox{[see
  (\ref{recursive})]} \nonumber \\
 \len(g_j^*) &=& {\cal O}( [\len(f_t) \cdot \len(g_t)]^2 ) \;\;\; \mbox{ for large }
\len(f_t) \cdot \len(g_t)   \label{len_g}
\end{eqnarray}
\noindent To solve the 3-CNF problem, one should compute all $n$
functional descriptors $h_t^{(\{l,\cdots,m\})}(\cdot)$ at each step of integration of
the $m$ clauses, i.e. for each $\varphi^{(\{l,\cdots,m\})}$ with $l$
decreasing from $m$ to $1$. Each  $h_t^{(\{l,\cdots,m\})}(\cdot)$
could yield to at most $n$
recursive calls with similar complexity.  So, using the equivalence between (\ref{len_h}) and (\ref{complexity}), the overall complexity of the functional
approach to 3-CNF problem will be of order 
{$\displaystyle {\cal
  O}( m \; n^2 \;\max_{1 \leq t \leq n} \; \max_{1 \leq l \leq m} \len(h^{(\{l,\cdots,m\})}_t) \; )$}
\end{proof}
\subsection{Cluster effect}
\noindent From numerical tests, we see that {\em a cluster effect} appears in the middle of the
algorithm.
This is understandable as the variables with smaller
indices are subject to more and more constraints coming from the first treated
clauses. When
$W^{(n-t+1)}(x_t) = \{\alpha_1, \cdots, \alpha_{t}\}$, we see a linear
decrease of $\# W^{(n-t+1)}(x_t)$ as $t$ decreases from $\sim \frac{n}{3}$ to $1$. This {\em
  cluster effect} is at the hart of {\em hard} 3-CNF-SAT problems. 
\mbox{}\\[12pt]
\noindent See numerical results
in {\em Figure 1} for the
3-CNF-SAT problem {\em uuf50-02.cnf} taken from {\it
  http://www.cs.ubc.ca/$\sim$hoos/SATLIB/benchm.html}.  Both graphs shows the complexity (in log scale and in normal scale) as the descriptor algorithm goes from the first clause to the last one (in the sorted 3-CNF-SAT problem).  \\
\begin{figure}[htb]
\centering
\includegraphics[height=5cm]{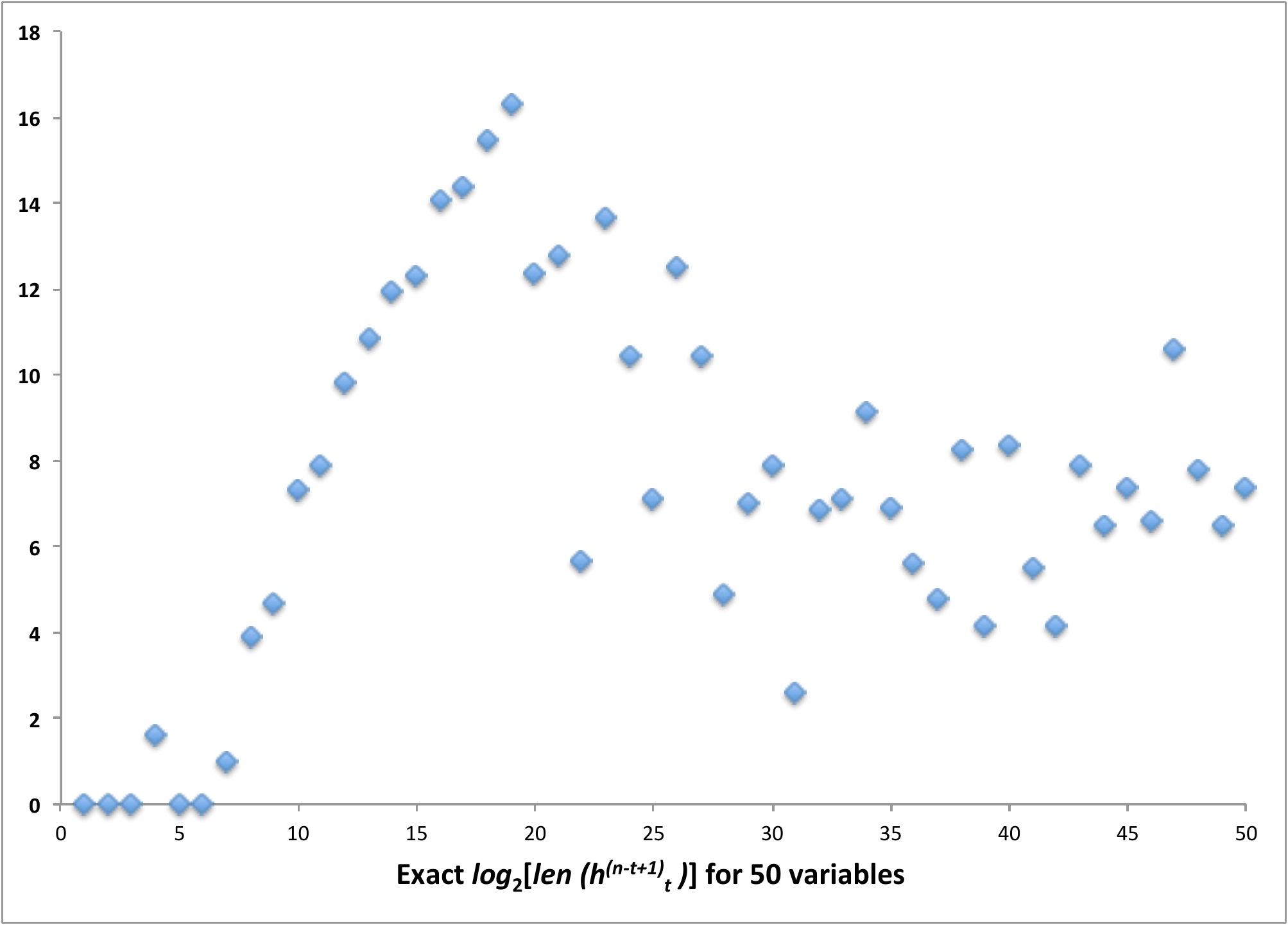}
\hspace{18pt} \includegraphics[height=5cm]{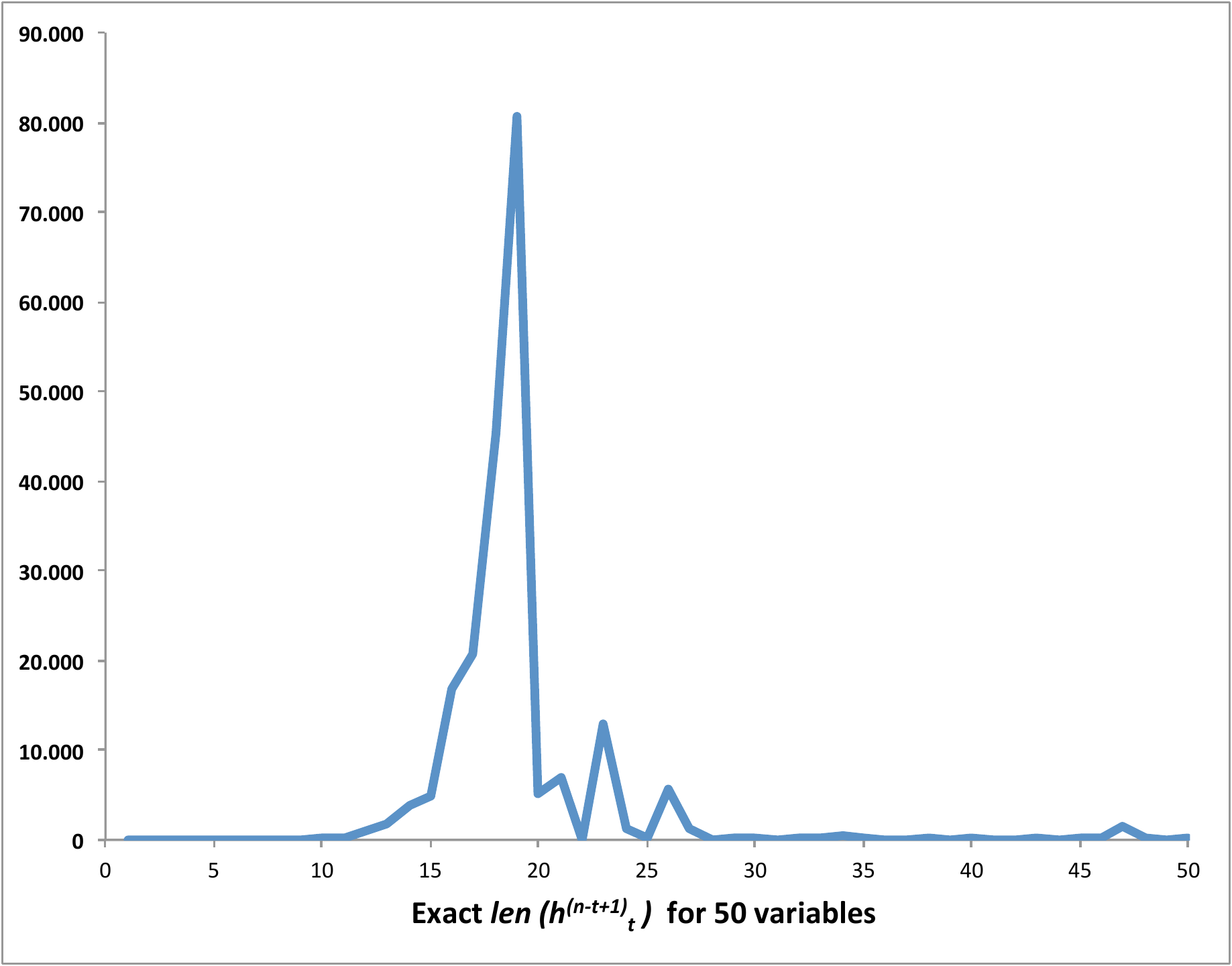} \\
{\small Figure 1 : $\log_2 \len(h_t(\cdot))$ and $\len(h_t(\cdot))$ for 3-CNF-SAT problem uuf50-02.cnf}
\end{figure}
%
\section{Complexity theorems for listing the solutions, given  ${\cal H}_{\varphi}$} 
\begin{theorem} {The complexity for {\bf \em listing} the solutions of a ``hard''
    3-CNF-SAT $\varphi$, given ${\cal
  H}_{\varphi}$, is {\bf \em polynomial}.} \label{complexity_3} \\[12pt]
\noindent We consider a ``hard'' 3-CNF-SAT problem $\varphi$ with $n$ propositional
variables. We suppose ${\cal
  H}_{\varphi}$ is computed and available.  Let  $\Sigma_\varphi = \#
\; {\cal S}_\varphi \; = 2^{{\cal O}(1)} \;$
 be the number of solutions for $\varphi$. Then the complexity needed
 to list all $\Sigma_\varphi$ solutions from ${\cal H}_{\varphi}$ is
 ${\cal O}(2 \; n \; \Sigma_\varphi) = n \; 2^{{\cal O}(1)} = {\cal O}(n)$.\\
\end{theorem} 
\begin{proof}
\mbox{}\\
\noindent Let us note that if $ {\cal S}_\varphi = \emptyset$
(no solution), ${\cal H}_{\varphi}$ does not exist.  If there is only
one solution $\{\bar{s}_1\}$, ${\cal H}_{\varphi}$ is a constant function, as $Im \;
{\cal H}_{\varphi} = \{\bar{s}_1\}$. \\[12pt]
Let  $\displaystyle \; {\cal S}_{\varphi} = \{\bar{s}_1, \cdots,
\bar{s}_{\Sigma_\varphi}\}$ be the set of solutions with $\bar{s}_j =
(s_j^1, \cdots, s_j^n)$ and $s_j^i \in \{0,1\} $. We can describe the
solutions as leafs of a tree.\\[12pt]
\begin{center}
\fbox{\synttree[[$h_1(\alpha_1) \equiv 0$[$h_2(0,\alpha_2) = 0$[$h_3(0,0,\alpha_3)=0$[$\bar{s}_1
=(0,0,0)$]][$h_3(0,0,\alpha_3)=1$[$\bar{s}_2
=(0,0,1)$]]][$h_2(0,\alpha_2) = 1$[{\tiny No
solution (0,1,0,$\cdots$)}][$h_3(0,1,\alpha_3) \equiv 1$[$\bar{s}_3
=(0,1,1)$]]]][{\tiny No
solution (1,$\cdots$)}]]}
{\small Figure 3.  Tree representation example of the solutions for $\varphi$}
\end{center}
\mbox{}\\
\noindent For each node, one needs to find whether $Im \;
h_t(\cdots,\alpha_t) = \{0,1\}, \; \{0\}$ or $\{1\}$, where $``\cdots"$
represents the branch to the node.  This takes ${\cal O}(2 \times \#$
nodes) operations.  \\[12pt]
\noindent Each solution corresponds to a leaf of the tree, and
the branch to it contains $n$ nodes.  So, the maximal number of nodes
for $\Sigma_\varphi$ solutions is $n \times \Sigma_\varphi$.
Therefore, the complexity for listing the solutions of a {\bf ``hard" 3-CNF-SAT}
problem $\varphi$ is ${\cal O}(2\; n\; \Sigma_\varphi) \; = \; {\cal O}(n)$ as $\Sigma_\varphi = 2^{{\cal O}(1)}$.\\
\end{proof}
\begin{theorem}{The complexity for {\bf \em listing} the solutions common to
    many ``solutions trees'' is bound by their minimal complexity.}
\end{theorem}
\begin{center}
\begin{figure}[htb]
\hspace{2cm} \includegraphics[height=6cm]{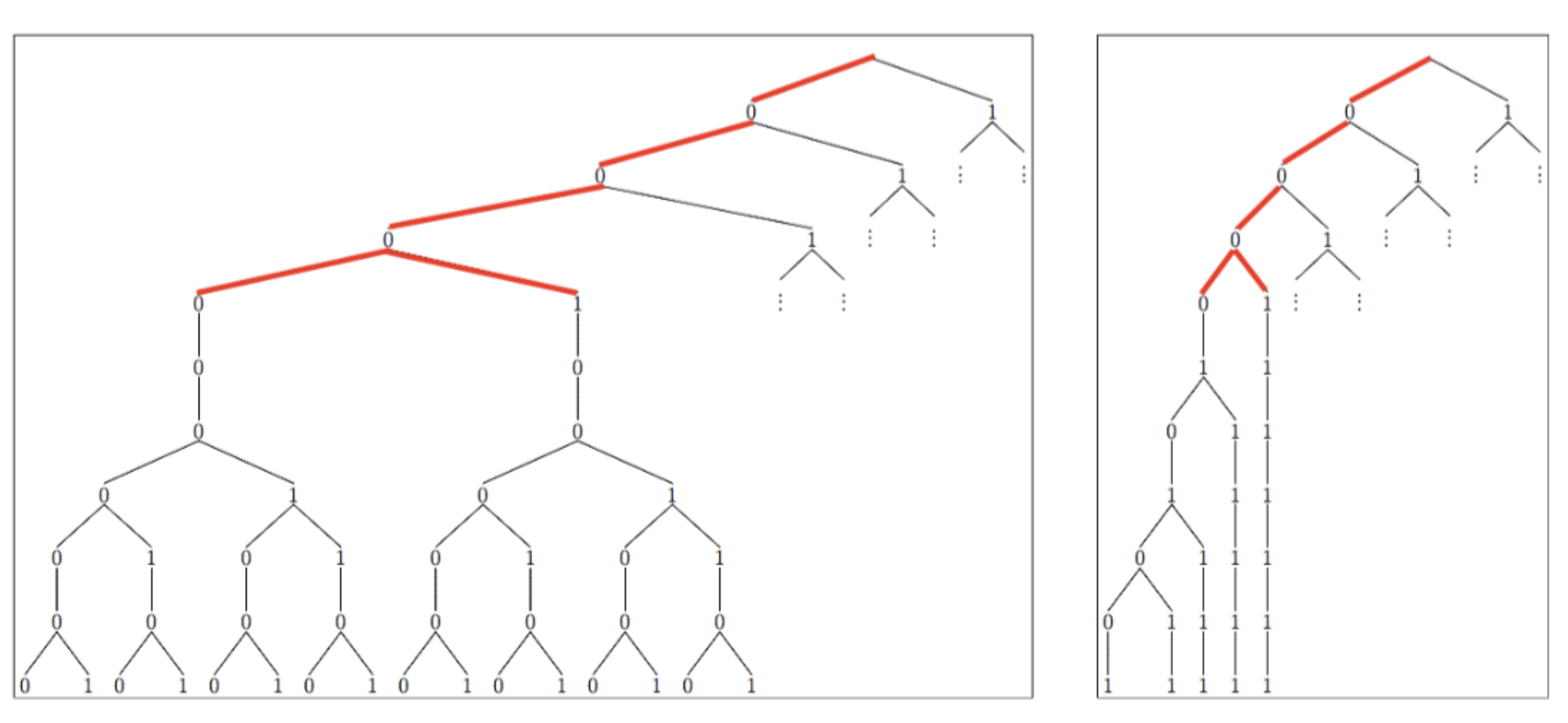} \\
\centerline{  {\small Figure 4 : Merging of two solutions trees
   $\Sigma_{\varphi^{(L)}}$ and $\Sigma_{\varphi^{(L')}}$, without solution for the {\em red} branches.}}\end{figure}
\end{center}
\begin{proof}\\
\noindent The complexity for listing the solutions of the sub-problem
$\varphi^{(L)}$  is ${\cal O}(2 \;n\; \Sigma_{\varphi^{(L)}}).$  To list
the solutions common to several ``solutions trees", one needs to follow
the paths in common (in red on figure 4).  The number of paths in common
will be less or equal to the minimum number of paths among the different
sub-problems, as the common paths should belong to this ``solutions tree''.
So, the complexity is ${\cal O}(2 \;n\;\;\min_{L} \; \Sigma_{\varphi^{(L)}})$.\\
\end{proof}
\section{The indicator function $\mathbbm{1}_{{\cal S}_\varphi}(x_1, \cdots,x_n)$ of the set of solutions ${\cal S}_\varphi$}
\subsection{Another description of ${\cal S}_{\varphi} = Im \; {\cal
  H}_\varphi$.}
\noindent Let consider a solution $(x_1, \cdots, x_n)$  for the 3-CNF-SAT problem $\varphi$ :
\begin{eqnarray}
 (x_1, \cdots, x_n) \in {\cal S}_{\varphi} 
&\Leftrightarrow &{\cal H}_{\varphi}(x_1, \cdots, x_n) 
 = (x_1, \cdots, x_n) \;\; \mbox{ by construction of } {\cal H}_{\varphi}  \nonumber  \\
& &  \nonumber \\
& \Leftrightarrow &  \raisebox{5ex}{t} \left[ \begin{array}{c}
h_{1,\varphi}(x_1)\\
\vdots\\
h_{n, \varphi}(x_1,\cdots,x_n)
\end{array}
\right] = \raisebox{5ex}{t} \left[ \begin{array}{c}
x_1\\
\vdots\\
x_n
\end{array}
\right]   \nonumber \\
& &  \nonumber \\
& \stackrel{\mbox{\tiny \em (mod 2)}}{\Leftrightarrow} &  \raisebox{5ex}{t} \left[ \begin{array}{c}
h_{1,\varphi}(x_1) + x_1\\
\vdots\\
h_{n, \varphi}(x_1,\cdots,x_n) + x_n
\end{array}
\right] = \raisebox{5ex}{t} \left[ \begin{array}{c}
0\\
\vdots\\
0
\end{array}
\right] \nonumber \\
& & \nonumber \\
& \stackrel{\mbox{\tiny \em (mod 2)}}\Leftrightarrow & \prod_{i=1}^{n} \; [ h_{i,\varphi}(x_1,\cdots,x_i) + x_i + 1 ] \; = \; 1 \label{indicator}
\end{eqnarray} 
\mbox{}\\
\begin{theorem} {\bf \em Indicator function $\mathbbm{1}_{{\cal S}_\varphi}(x_1, \cdots,x_n)$ and descriptor function ${\cal H}_{\varphi}$}
\mbox{} \\
Consider the 3-CNF-SAT problem $\varphi$ and its descriptor function ${\cal H}_{\varphi}$, if it exists.  The {\em indicator function} of the set of solutions ${\cal S}_\varphi$ is given by :
\begin{eqnarray}
\mathbbm{1}_{{\cal S}_\varphi}(x_1, \cdots,x_n) \equiv 
\left\{ \begin{array}{ll}
\displaystyle \prod_{i=1}^{n} \; [ h_{i,\varphi}(x_1,\cdots,x_i) + x_i + 1 ] \;
\;\; & \mbox{ if ${\cal H}_\varphi$ exists} \\ 
0 \;\;\; & \mbox{ otherwise} 
\end{array} \right. \label{indicator2}
\end{eqnarray}
\end{theorem}
\mbox{}\\
\begin{proof}
\mbox{}\\
If ${\cal H}_{\varphi}$ does not exist, this means that there is no solution for $\varphi$ : ${\cal S}_\varphi = \emptyset $.  \\
From (\ref{indicator}), it follows that :
\[  (x_1, \cdots, x_n) \not \in {\cal S}_{\varphi}  \stackrel{\mbox{\tiny \em (mod 2)}}\Leftrightarrow  \prod_{i=1}^{n} \; [ h_{i,\varphi}(x_1,\cdots,x_i) + x_i + 1 ]  \; = \; 0 \]
So, the definition given in (\ref{indicator2}) corresponds to the indicator function of ${\cal S}_\varphi$.
\end{proof}
\subsection{Properties}
\noindent $\bullet$ Let $\varphi$ and $\varphi'$ be two 3-CNF-SAT problems, and ${\cal S}_\varphi \cap {\cal S}_{\varphi'}$ the set of common solutions. Then, ${\cal S}_{\varphi \wedge \varphi'} = {\cal S}_\varphi \cap {\cal S}_{\varphi'}$ and $\mathbbm{1}_{{\cal S}_{\varphi \wedge \varphi'}}(\cdot) =  \mathbbm{1}_{{\cal S}_\varphi}(\cdot) \times \mathbbm{1}_{{\cal S}_{\varphi'}}(\cdot)$, following the normal properties of indicator functions.  \\[12pt]
\noindent $\bullet$ Let ${\cal S}_\varphi = \{(s_1, \cdots, s_n)\}$  [$\varphi$ has only one solution] then 
\begin{eqnarray}
\mathbbm{1}_{{\cal S}_\varphi}(x_1, \cdots, x_n) = \prod_{i=1}^n 
x_i^{s_i} (x_i + 1)^{(s_i + 1)} \;\;\; \mbox{\tiny \em (mod 2)}  \label{indicator5}
\end{eqnarray}

as $\mathbbm{1}_{{\cal S}_{\varphi}} (x_1, \cdots, x_n) = 1 \;\; \iff \;\; \forall i \; : \; x_i^{s_i} (x_i + 1)^{(s_i+1)} = 1 \; \; \iff \;\; \forall i \; : \; x_i = s_i.$ \\[12pt]
\noindent $\bullet$ Let ${\cal S}_\varphi$ have many solutions, then 
\begin{eqnarray}
\mathbbm{1}_{{\cal S}_\varphi}(x_1, \cdots, x_n) = \sum_{(s_1, \cdots,s_n) \in {\cal S}_\varphi} \;  \prod_{i=1}^n 
x_i^{s_i} (x_i + 1)^{(s_i + 1)} \;\;\; \mbox{\tiny \em (mod 2)}  \label{indicator4}
\end{eqnarray} 
The proof is easily done by recurrence. \\[12pt]
\noindent $\bullet$ Let  $\varphi$ be a  3-CNF formula.
\begin{eqnarray}
\varphi &=& \bigwedge_{k=1}^m \psi_k \nonumber \\
\mbox{ where  }  & \;& \psi_k = [\neg]^{\delta_{r_k}} \; x_{r_k} \vee [\neg]^{\delta_{s_k}} \; x_{s_k} \vee [\neg]^{\delta_{t_k}} \;  x_{t_k}
\;\; \mbox{ where } [\neg]^{\delta_i} x_i =
\left\{ \begin{array}{r}
 x_i \; \mbox{ if } \delta_i = 0  \\
 \neg \; x_i \; \mbox{ if } \delta_i = 1 
\end{array} \right. \nonumber \\
\mbox{Then } &\; & \nonumber \\
\mathbbm{1}_{{\cal S}_\varphi}(\cdot) &=& \prod_{k=1}^m \; \{ x_{r_k}^{\delta_{r_k}} \; (x_{r_k}+1)^{(\delta_{r_k}+1)} \;
x_{s_k}^{\delta_{s_k}} \; (x_{s_k}+1)^{(\delta_{s_k}+1)} \;
x_{t_k}^{\delta_{t_k}} \; (x_{t_k}+1)^{(\delta_{t_k}+1)} \; + \; 1 \} \label{indicator3}
\end{eqnarray}
The proof follows directly from (\ref{def_h}) and (\ref{indicator}), considering that :
\begin{eqnarray}
\mathbbm{1}_{{\cal S}_\varphi}(x_1, \cdots, x_n) = \prod_{k=1}^m \mathbbm{1}_{{\cal S}_{\psi_k}}(x_1, \cdots, x_n) \label{property5}
\end{eqnarray} 
Let us remark that (\ref{indicator3}) could be of exponential complexity, as it is a multi-linear product of $m$ sums of at least two terms, with $m = {\cal O}(n)$.  Numerical results show here again a {\em cluster effect} when computing (\ref{indicator3}).   \\
\section{A greedy polynomial algorithm for ``hard" 3-CNF-SAT problems}
\subsection{The sub-problems $\varphi_{(x_t)}^{\oplus}$ and $\varphi_{(x_t)}^{\ominus}$}
\begin{definition}
Let $\varphi_{(x_t)}^{\oplus}$ be the sorted sub-problem of $\varphi$, restricted to the clauses in $Cl(x_t)$ having $x_t$ as the highest indexed positive variable :
\[\varphi_{(x_t)}^{\oplus} = \bigwedge_{k=1}^{m_t^{\oplus}} [\neg] x_{r_k} \vee [\neg] x_{s_k} \vee x_{t} \;\; = \; \bigwedge_{k=1}^{m_t^\oplus} \psi_k  \;\; \mbox{where } r_k < s_k < t \]
and $\varphi_{x_t}^{\ominus}$ with the clauses in $Cl(x_t)$ having $x_t$ as the highest indexed negative variable :
\[\varphi_{(x_t)}^{\ominus} = \bigwedge_{k=1}^{m_t^{\ominus}} [\neg] x_{r_k} \vee [\neg] x_{s_k} \vee \neg \; x_{t_k} \;\; = \; \bigwedge_{k=1}^{m_t^\ominus} \psi_k  \;\; \mbox{where } r_k < s_k < t \]
\end{definition} 
\begin{eqnarray*}
\hspace{-5cm} \mbox{We get } \varphi = \bigwedge_{t=1}^n [\; \varphi_{(x_t)}^{\oplus} \; \wedge \; \varphi_{(x_t)}^{\ominus} \;] \; \mbox{ and } \;
m = \sum_{t=1}^n \; ( m_t^{\oplus} + m_t^{\ominus} ) \\
\end{eqnarray*}
Considering together the definition (\ref{property5}) of $\mathbbm{1}_{{\cal S}_\varphi}(x_1, \cdots,x_n)$ 
and (\ref{indicator2}), we get :
\begin{eqnarray}
\mathbbm{1}_{{\cal S}_\varphi}(x_1, \cdots, x_n) & = & \prod_{t=1}^n [ \mathbbm{1}_{{\cal S}_{\varphi_{(x_t)}^\oplus}(x_1, \cdots, x_n)} 
\; \cdot \; \mathbbm{1}_{{\cal S}_{\varphi_{(x_t)}^\ominus}(x_1, \cdots, x_n)} ] \nonumber \\
& = & \prod_{t=1}^n \; \prod_{i=1}^{t} \; [ h_{i,\varphi_{(x_t)}^\oplus}(x_1,\cdots,x_i) + x_i + 1 ] \; \cdot \; 
 \prod_{t=1}^n \; \prod_{i=1}^{t} \; [ h_{i,\varphi_{(x_t)}^\ominus}(x_1,\cdots,x_i) + x_i + 1 ]  \nonumber  \\
 & \stackrel{(\ref{property2})}{=} & \prod_{t=1}^n \; [ h_{t,\varphi_{(x_t)}^\oplus}(x_1,\cdots,x_t) + x_t + 1 ] \; \cdot \; 
 \prod_{t=1}^n \; [ h_{t,\varphi_{(x_t)}^\ominus}(x_1,\cdots,x_t) + x_t + 1 ] \label{indicator6}
\end{eqnarray}
where $\len(h_{t,\varphi_{(x_t)}^\oplus}(x_1,\cdots,x_t) + x_t + 1) = {\cal O}(2^\Delta)$ and $\len(h_{t,\varphi_{(x_t)}^\ominus}(x_1,\cdots,x_t) + x_t + 1) = {\cal O}(2^\Delta)$.  See (\ref{Delta}).
\mbox{}\\
\begin{theorem} {\em The computation of $\; {\cal H}_{\varphi_{(x_t)}^\oplus}$ and $\; {\cal H}_{\varphi_{(x_t)}^\ominus}$ is ${\cal O}(n^k)$ for 3-CNF-SAT problems where $\Delta = \frac{m}{n} = {\cal O}(1)$. }
\end{theorem}
\mbox{}\\
\begin{proof}
\noindent For each $\psi_k$ in $\varphi_{(x_t)}^{\oplus}$, we get $h_{t,\psi_k}(\alpha_{r_k},\alpha_{s_k},1) = h_{t,\psi_k}(\alpha_1, \cdots ,\alpha_{t-1},1) = 1$  and $h_{i,\psi_k}(\alpha_1,\alpha_2,\cdots,\alpha_i) = \alpha_i \; $ for $ \; i < t$. See (\ref{def_h}). \\
Remember the computation algorithm (see {\em Theorem \ref{computation_algo}}). By (\ref{merge}), we get : 
\begin{eqnarray*}
h_{t,\psi_1 \wedge \psi_2}(\alpha_1,\cdots, \alpha_t) := 
& (\alpha_t +1) &\; \cdot \;  \{ \; [h_{t,\psi_1}(\alpha_1,\cdots, \alpha_{t-1},0) +  
 h_{t,\psi_2}(\alpha_1,\cdots, \alpha_{t-1},0)]  \\
& &  \; \; \;  \; \;\cdot \; [ h_{t,\psi_1}(\alpha_1,\cdots, \alpha_{t-1},1) \cdot  
 h_{t,\psi_2}(\alpha_1,\cdots, \alpha_{t-1},1)]   \\
& &  \; \; + \; [h_{t,\psi_1}(\alpha_1,\cdots, \alpha_{t-1},0) \cdot  
 h_{t,\psi_2}(\alpha_1,\cdots, \alpha_{t-1},0)] \; \}  \\
&  \; +\;\;  \; \alpha_t &\; \cdot \;
   \{ \; [h_{t,\psi_1}(\alpha_1,\cdots, \alpha_{t-1},1) +  
 h_{t,\psi_2}(\alpha_1,\cdots, \alpha_{t-1},1)]   \\
& &  \; \;\;  \; \;\cdot \; [h_{t,\psi_1}(\alpha_1,\cdots, \alpha_{t-1},0) +  
 h_{t,\psi_2}(\alpha_1,\cdots, \alpha_{t-1},0) ]   \\
& &  \; \; \; + \; [h_{t,\psi_1}(\alpha_1,\cdots, \alpha_{t-1},1) +  
 h_{t,\psi_2}(\alpha_1,\cdots, \alpha_{t-1},1)]    \\
& &  \;  \; \; \;  \; \cdot \; [h_{t,\psi_1}(\alpha_1,\cdots, \alpha_{t-1},0) \cdot  
 h_{t,\psi_2}(\alpha_1,\cdots, \alpha_{t-1},0)]    \\
& & \;  \; \;+   \; [h_{t,\psi_1}(\alpha_1,\cdots, \alpha_{t-1},1) \cdot  
 h_{t,\psi_2}(\alpha_1,\cdots, \alpha_{t-1},1) ] \} \\
:=  & (\alpha_t +1) &\; \cdot \;  \{ \; [h_{t,\psi_1}(\alpha_1,\cdots, \alpha_{t-1},0) +  
 h_{t,\psi_2}(\alpha_1,\cdots, \alpha_{t-1},0)]  \\
& &  \; \; + \; [h_{t,\psi_1}(\alpha_1,\cdots, \alpha_{t-1},0) \cdot  
 h_{t,\psi_2}(\alpha_1,\cdots, \alpha_{t-1},0)] \; \}  \\
&  \; +\;\;  \; \alpha_t &\; \; \cdot \; 1 \\
&\hspace{-9cm} \mbox{and }& \\
 h_{i,\psi_1 \wedge \psi_2}(\alpha_1,\cdots, \alpha_{i}) := & \alpha_i & \; \; \; \forall \; i < t \;\;\;\; \mbox{ as } 
h_{i,\psi_1}(\alpha_1,\cdots, \alpha_i) = h_{i,\psi_2}(\alpha_1,\cdots, \alpha_i) = \alpha_i
\end{eqnarray*}
\begin{eqnarray}
\stackrel{\mbox{\tiny \em (After step one)}}{\Longrightarrow} \;\;\; \left\{
\begin{array}{l}
h_{t,\psi_1 \wedge \psi_2}(\alpha_1,\cdots, \alpha_{t-1}, 1) = 1 \;\; \; \forall  (\alpha_1, \cdots, \alpha_{t-1}) \in \{0,1\}^{t-1} \\
h_{i,\psi_1 \wedge \psi_2}(\alpha_1,\cdots, \alpha_i) = \alpha_i  \;\; \forall \; i \neq t \;\; \; \forall  (\alpha_1, \cdots, \alpha_i) \in \{0,1\}^i 
\end{array} \right. \label{property1}
\end{eqnarray}  
For the next step of the algorithm [the computation of $h_{t,(\psi_1 \wedge \psi_2)\; \wedge \psi_3}(\alpha_1,\cdots, \alpha_t)$], we will get the same property (\ref{property1}) as $h_{t,\psi_1 \wedge \psi_2}(\alpha_1,\cdots, \alpha_{t})$
 will replace $h_{t,\psi_1}(\alpha_1,\cdots, \alpha_t)$ in the formula (\ref{merge}) and $h_{t,\psi_3}(\alpha_1,\cdots, \alpha_{t})$ the term $h_{t,\psi_2}(\alpha_1,\cdots, \alpha_t)$.

\noindent Let us note that there will be no recursive call in the algorithm, as the condition (\ref{recursive}) never occurs.  Indeed :
\[ [h_{t,\psi_1}(\alpha_{r_1},\alpha_{s_1},0) +  
 h_{t,\psi_2}(\alpha_{r_2},\alpha_{s_2},0)]  \; \cdot \;
 [h_{t,\psi_1}(\alpha_{r_1},\alpha_{s_1},1) +  
 h_{t,\psi_2}(\alpha_{r_2},\alpha_{s_2},1)] \; \neq 1 \] 
as $h_{t,\psi_1}(\alpha_{r_1},\alpha_{s_1},1) = 1$ and $h_{t,\psi_2}(\alpha_{r_2},\alpha_{s_2},1)=1$.  

\noindent Therefore, at the end of the algorithm, we get ${\cal H}_{\varphi_{(x_t)}}^{\oplus} = [h_{1,\varphi_{(x_t)}^{\oplus}}(\cdot) \hspace{6pt} \cdots \hspace{6pt} h_{t,\varphi_{(x_t)}^{\oplus}}(\cdot)] $  : 
\begin{eqnarray}
\mbox{with  } \;\; \left\{ \begin{array}{l}
h_{t,\varphi_{(x_t)}^{\oplus}}(\alpha_1,\cdots, \alpha_{t-1}, 1) = 1 \;\; \; \forall  (\alpha_1, \cdots, \alpha_{t-1}) \in \{0,1\}^{t-1} \\
h_{i,\varphi_{(x_t)}^{\oplus}}(\alpha_1,\cdots, \alpha_i) = \alpha_i  \;\; \forall \; i \neq t \;\; \; \forall  (\alpha_1, \cdots, \alpha_i) \in \{0,1\}^i 
\end{array} \right. \label{property2}
\end{eqnarray}

\noindent Similarly, we  get ${\cal H}_{\varphi_{(x_t)}}^{\ominus} = [h_{1,\varphi_{(x_t)}^{\ominus}}(\cdot) \hspace{6pt} \cdots \hspace{6pt} h_{t,\varphi_{(x_t)}^{\ominus}}(\cdot)] $ : 
\begin{eqnarray}
\mbox{with  } \;\; \left\{ \begin{array}{l}
h_{t,\varphi_{(x_t)}^{\ominus}}(\alpha_1,\cdots, \alpha_{t-1}, 0) = 0 \;\; \; \forall  (\alpha_1, \cdots, \alpha_{t-1}) \in \{0,1\}^{t-1} \\
h_{i,\varphi_{(x_t)}^{\ominus}}(\alpha_1,\cdots, \alpha_i) = \alpha_i  \;\; \forall \; i \neq t \;\; \; \forall  (\alpha_1, \cdots, \alpha_i) \in \{0,1\}^i 
\end{array} \right. \label{property3}
\end{eqnarray}

\noindent Theorem \ref{complexity_1} states that the complexity to compute ${\cal H}_{\varphi_{(x_t)}}^{\oplus}$ is :
\begin{eqnarray}
{\cal
  O}( m_t^\oplus \; n^2 \; \max_{1 \leq i \leq t} \;  \len(h_{i,\varphi_{(x_t)}^{\oplus}}(\cdot)) \; ) =
  {\cal
  O}( m_t^\oplus \; n^2 \; \len(h_{t,\varphi_{(x_t)}^{\oplus}}(\alpha_1,\cdots, \alpha_{t-1}, 0)) \; ) \label{property4}
\end{eqnarray}
as $\forall i < t : \len(h_{i,\varphi_{(x_t)}^{\oplus}}(\alpha_1,\cdots, \alpha_{i})) = \len(\alpha_i) =  \len(h_{t,\varphi_{(x_t)}^{\oplus}}(\alpha_1,\cdots, \alpha_{t-1}, 1)) = \len(1) =  1$.
\mbox{}\\[12pt]
\noindent All together, the $m_t^\oplus$ clauses of $\varphi_{(x_t)}^\oplus$ concern at most $(2 \cdot m_t^\oplus) + 1$ variables.
Therefore, 
\begin{eqnarray}
\len(h_{t,\varphi_{(x_t)}^{\oplus}}(\alpha_1,\cdots, \alpha_{t-1}, 0)) \leq 2^{(2 \cdot m_t^\oplus) + 1} \; = \;
{\cal O}(2^{\Delta + 1}) \label{Delta}
\end{eqnarray}
where $\Delta = \frac{m}{n} $ is the {\em ratio} of the 3-CNF-SAT problem, see (\ref{no1}).  
Without loss of generality, one can {\bf relabel} the variables so that the minimal index $1$ is attributed to the most frequent variable and so on
for the remaining variables.  An exact uniform distribution of the variables yields to $m_t^\oplus + m_t^\ominus \leq 3 \Delta$ for all $t$. An extreme non uniform distribution, i.e. when each variable occurs only once except for one variable, relabeled $x_1$, that occurs
$3 m - (n - 1)$, yields to  $m_t^\oplus + m_t^\ominus = 1$ for all $t$. 

If the variables are at random in $\varphi$, $\neg x_t$ and $x_t$ should occur approximately $\frac{1}{2}\frac{m}{n}$ times, and $m_t^\oplus \approx m_t^\ominus \approx \frac{\Delta}{2}$ for large $n$.  This ratio is a good indicator of the hardness of the 3-CNF-SAT problem (see \cite{Crawford199631}).  

\noindent $\Delta$ being a constant with respect to $n$, the complexity for 
${\cal H}_{\varphi_{(x_t)}}^{\oplus}$ given in (\ref{property4}) is thus :
\begin{eqnarray}
{\cal O}(\frac{\Delta}{2} \cdot n^2 \cdot 2^\Delta) = {\cal O}(n^2) \label{property6}
\end{eqnarray}

\noindent The same arguments are used to prove that the complexity to get ${\cal H}_{\varphi_{(x_t)}}^{\ominus}$ is ${\cal O}(n^2)$.
\end{proof}
\subsection{Complexity theorem for 3-CNF-SAT problem $\varphi$ with one or zero solution}
\begin{theorem}{\bf Necessary and sufficient condition for satisfiability \linebreak when $ \# {\cal S}_\varphi \leq 1$} \\
  Let $ \# {\cal S}_\varphi \; \leq 1$, then $\varphi$ is satisfiable [$\# {\cal S}_\varphi \; = 1$] if and only if
\begin{eqnarray}
 \prod_{t=1}^n 
\; [ h_{t,\varphi_{(x_t)}^\oplus}(x_1,\cdots,x_t) + x_t + 1 ] \;& \cdot& \; 
 \prod_{t=1}^n \; [ h_{t,\varphi_{(x_t)}^\ominus}(x_1,\cdots,x_t) + x_t + 1 ] \nonumber \\ 
 & = & \left( \prod_{t = 1}^n x_t  \right) + {\cal E}(x_1,\cdots,x_n)  \label{satis1} \\
 \nonumber \\
 \mbox{where  } \;\; {\cal E}(x_1,\cdots,x_n) \;\;& \mbox{ is}&\mbox{ a linear combination of rank strictly less than $n$} \nonumber
\end{eqnarray}
\end{theorem}
\mbox{}\\
\begin{proof}
From the hypothesis, $ \varphi$ is satisfiable if and only if $\# {\cal S}_\varphi \; = 1 $. Let $(s_1, \cdots, s_n)$ be this only solution of $\varphi$. Thus,
\begin{eqnarray*}
\mathbbm{1}_{{\cal S}_\varphi}(x_1, \cdots, x_n)  & \stackrel{(\ref{indicator6})}{=} & \prod_{t=1}^n 
\; [ h_{t,\varphi_{(x_t)}^\oplus}(x_1,\cdots,x_t) + x_t + 1 ] \; \cdot \; 
 \prod_{t=1}^n \; [ h_{t,\varphi_{(x_t)}^\ominus}(x_1,\cdots,x_t) + x_t + 1 ] \\
 & \stackrel{(\ref{indicator5})}{=} & \prod_{t=1}^n 
x_t^{s_t} (x_t + 1)^{(s_t + 1)} = \prod_{t \dagger s_t=1} (x_t) \cdot \prod_{t \dagger s_t=0} (x_t+1) \\
& = & \left( \prod_{t \dagger s_t=1} x_t  \right) \cdot \left[ \left( \prod_{t \dagger s_t=0} x_t  \right) + \sum_{j \dagger s_j=0} \left( \prod_{t \neq j \dagger s_t=0} x_t  \right) \; + \; \cdots \right] \\
 & = & \left( \prod_{t = 1}^n x_t  \right) + {\cal E}(x_1,\cdots,x_n)
\end{eqnarray*}
${\cal E}(x_1,\cdots,x_n)$ depends on the value of $(s_1,\cdots,s_n) : {\cal E}(x_1,\cdots,x_n) = 0$ when $(s_1,\cdots,s_n) = (1, \cdots,1)$ and ${\cal E}(x_1,\cdots,x_n) =  \prod_{t=1}^n (x_t + 1) - \prod_{t=1}^n (x_t) $ when $(s_1,\cdots,s_n) = (0, \cdots,0)$. Whatever the solution, $(\prod_{t = 1}^n x_t)$ appears in (\ref{indicator6}). \\
As
\[ [ {\cal S}_\varphi \; = \emptyset ] \;\; \Leftrightarrow \;\; \prod_{t=1}^n \;
[ h_{t,\varphi_{(x_t)}^\oplus}(x_1,\cdots,x_t) + x_t + 1 ] \; \cdot \;
\prod_{t=1}^n \; [ h_{t,\varphi_{(x_t)}^\ominus}(x_1,\cdots,x_t) + x_t + 1 ] \equiv 0 \]
then, the apparition of $\displaystyle (\prod_{t = 1}^n x_t)$  in (\ref{indicator6})  is equivalent to prove that $\# {\cal S}_\varphi \; = 1$.\\ 
\end{proof}
\mbox{}\\
\noindent {\bf Example :}
Let us consider the following 3-CNF-SAT problem with 6 variables and no solution :
\begin{eqnarray*}
\varphi &= & (x_6 \vee x_2 \vee x_1) 
  \wedge  (x_6 \vee x_3  \vee \neg \;x_2) \\ 
 & \wedge & (\neg \;x_6 \vee \neg \;x_5  \vee x_1) 
  \wedge  (\neg \;x_6 \vee \neg \;x_4  \vee x_1) 
  \wedge  (\neg \;x_6 \vee \neg \;x_3  \vee \neg \;x_1) 
  \wedge  (\neg \;x_6 \vee x_3  \vee \neg \;x_2) \\
 & \wedge & (x_5 \vee x_4  \vee x_1) 
  \wedge  (x_5 \vee \neg \;x_4  \vee \neg \;x_1)  
  \wedge  (x_5 \vee \neg \;x_3  \vee \neg \;x_1) \\ 
 & \wedge & (\neg \;x_5 \vee x_2  \vee \neg \;x_1) 
  \wedge  (\neg \;x_5 \vee x_4  \vee \neg \;x_1) 
  \wedge  (\neg \;x_5 \vee x_3  \vee \neg \;x_1) \\
 & \wedge & (x_4 \vee \neg \;x_2  \vee x_1)  
  \wedge  (\neg \;x_4 \vee \neg \;x_2  \vee x_1) 
  \wedge  (\neg \;x_4 \vee \neg \;x_3  \vee \neg \;x_2) \\
 & \wedge & (x_3 \vee x_2  \vee \neg \;x_1) 
 \wedge  (\neg \;x_3 \vee x_2  \vee \neg \;x_1) \\
 \end{eqnarray*}
This gives us :
\begin{eqnarray*}
h_{1,\varphi_{(x_1)}}^\oplus(\cdot) & = & x_1 \\
h_{2,\varphi_{(x_2)}}^\oplus(\cdot) & = & x_2 \\
h_{3,\varphi_{(x_3)}}^\oplus(\cdot) & = & x_1 + x_3 + x_1 x_2 + x_1 x_3 + x_1 x_2 x_3 \\
h_{4,\varphi_{(x_4)}}^\oplus(\cdot) & = & x_2 + x_4 + x_1 x_2 + x_2 x_4 + x_1 x_2 x_4 \\
h_{5,\varphi_{(x_5)}}^\oplus(\cdot) & = & 1 + x_1 + x_4 + x_1 x_3 + x_1 x_5 + x_4 x_5 + x_1 x_3 x_4 + x_1 x_3 x_5 + x_1 x_3 x_4 x_5 \\
h_{6,\varphi_{(x_6)}}^\oplus(\cdot) & = & 1 + x_1 + x_1 x_2 + x_1 x_6 + x_2 x_3 + x_1 x_2 x_6 + x_2 x_3 x_6  \\
\mbox{and} & & \\
h_{1,\varphi_{(x_1)}}^\ominus(\cdot) & = & x_1 \\
h_{2,\varphi_{(x_2)}}^\ominus(\cdot) & = & x_2 \\
h_{3,\varphi_{(x_3)}}^\ominus(\cdot) & = & x_3 + x_1 x_3 + x_1 x_2 x_3 \\
h_{4,\varphi_{(x_4)}}^\ominus(\cdot) & = & x_4 + x_2 x_4 + x_1 x_2 x_4 + x_1 x_2 x_3 x_4  \\
h_{5,\varphi_{(x_5)}}^\ominus(\cdot) & = & x_5 + x_1 x_5 + x_1 x_2 x_3 x_4 x_5 \\
h_{6,\varphi_{(x_6)}}^\ominus(\cdot) & = & x_6 + x_5 x_6 + x_4 x_6 + x_4 x_5 x_6 + x_2 x_6 + x_2 x_5 x_6 + x_2 x_4 x_6 + 
 x_2 x_4 x_5 x_6 + x_2 x_3 x_6 + \\
 & & x_2 x_3 x_5 x_6 + x_2 x_3 x_4 x_6 + x_2 x_3 x_4 x_5 x_6 + x_1 x_5 x_6 + x_1 x_4 x_6 + x_1 x_4 x_5 x_6 + x_1 x_3 x_6 + \\
 & & x_1 x_2 x_5 x_6 + x_1 x_2 x_4 x_6 + x_1 x_2 x_4 x_5 x_6 + x_1 x_2 x_3 x_5 x_6 + x_1 x_2 x_3 x_4 x_6 + x_1 x_2 x_3 x_4 x_5 x_6  \\
\end{eqnarray*}
\begin{eqnarray*}
\mathbbm{1}_{{\cal S}_\varphi}(&x_1&,\cdots, x_n)   =  \prod_{t=1}^n 
\; [ h_{t,\varphi_{(x_t)}^\oplus}(x_1,\cdots,x_t) + x_t + 1 ] \; \cdot \; 
 \prod_{t=1}^n \; [ h_{t,\varphi_{(x_t)}^\ominus}(x_1,\cdots,x_t) + x_t + 1 ] \\
& = & 2 \; x_1
+ 16 \; x_1 x_2
+ 10  \; x_1 x_3
+ 2  \; x_1 x_4
+ 10  \; x_1 x_5
+ 4  \; x_1 x_6
+ 2  \; x_4 x_6
+ 2  \; x_5 x_6
+ 242  \; x_1 x_2 x_3 \\
& + & 160  \; x_1 x_2 x_4
+ 80  \; x_1 x_2 x_5
+ 68  \; x_1 x_2 x_6
+ 10  \; x_1 x_3 x_4
+ 38  \; x_1 x_3 x_5
+ 56  \; x_1 x_3 x_6
+ 6  \; x_1 x_4 x_5 \\
& + & 26  \; x_1 x_4 x_6
+ 90  \; x_1 x_5 x_6
+ 6  \; x_2 x_3 x_4
+ 2 \;  x_2 x_3 x_5
+ 24  \;  x_2 x_4 x_6
+ 6  \; x_2 x_5 x_6
+ 8 \;  x_4 x_5 x_6 \\
&+& 3 \; 152  \; x_1 x_2 x_3 x_4 
+ 950  \; x_1 x_2 x_3 x_5
+ 2 \;  122  \; x_1 x_2 x_3 x_6
+ 624 \;  x_1 x_2 x_4 x_5
+ 2 \;  396  \; x_1 x_2 x_4 x_6 \\
&+& 1 \;  352 \;  x_1 x_2 x_5 x_6
+ 30  \; x_1 x_3 x_4 x_5
+ 176 \;  x_1 x_3 x_4 x_6
+ 508 \;  x_1 x_3 x_5 x_6
+ 268 \;  x_1 x_4 x_5 x_6 \\
&+& 10  \; x_2 x_3 x_4 x_5
+ 76 \;  x_2 x_3 x_4 x_6
+ 26 \;  x_2 x_3 x_5 x_6
+ 102 \;  x_2 x_4 x_5 x_6
+ 18 \;  950 \;  x_1 x_2 x_3 x_4 x_5 \\
&+& 75 \;  450  \; x_1 x_2 x_3 x_4 x_6
+ 25 \;  916  \; x_1 x_2 x_3 x_5 x_6
+ 26 \;  252  \; x_1 x_2 x_4 x_5 x_6
+ 1 \;  344  \; x_1 x_3 x_4 x_5 x_6 \\
&+& 384  \; x_2 x_3 x_4 x_5 x_6
+ 1 \;  086 \;  442 \;  x_1 x_2 x_3 x_4 x_5 x_6 \stackrel{mod \; 2}{=} 0  
\end{eqnarray*}
\noindent Let us consider the similar 3-CNF-SAT problem $\varphi'$ without the $13^{\mbox{th}}$ clause : 
$(x_4 \vee \neg \;x_2  \vee x_1)$.  $\varphi'$ has a unique solution : $(s_1, \cdots, s_6)=(0,1,1,0,1,0)$.
Here, we get :
\begin{eqnarray*}
\mathbbm{1}_{{\cal S}_\varphi}(x_1, &\cdots,& x_n)   =  \prod_{t=1}^n 
\; [ h_{t,\varphi_{(x_t)}^\oplus}(x_1,\cdots,x_t) + x_t + 1 ] \; \cdot \; 
 \prod_{t=1}^n \; [ h_{t,\varphi_{(x_t)}^\ominus}(x_1,\cdots,x_t) + x_t + 1 ] \\
& \stackrel{mod \; 2}{=} &  1 \;  x_2 x_3 x_5
+ 291  \; x_1 x_2 x_3 x_5
+ 3  \; x_2 x_3 x_4 x_5
+ 13 \;  x_2 x_3 x_5 x_6
+ 3\;633 \;  x_1 x_2 x_3 x_4 x_5 \\
&&+ 8 \;  293  \; x_1 x_2 x_3 x_5 x_6
+ 125  \; x_2 x_3 x_4 x_5 x_6
+ 257 \;  641 \;  x_1 x_2 x_3 x_4 x_5 x_6 \\
& \stackrel{mod \; 2}{=} & (1+ x_1)\; x_2 \; x_3 \; (1+x_4)\; x_5 \; (1+x_6) \\
\end{eqnarray*}
\begin{theorem} \label{big_thm}
{\em The 3-CNF-SAT problems, with $\#{\cal S} \leq 1$ and  $\Delta = \frac{m}{n} = {\cal O}(1)$, are in \classP}  
\end{theorem}
\mbox{}\\
{\em Proof : } \\
\raisebox{.5pt}{\textcircled{\raisebox{-.9pt} {1}}}    
$h_{t,\varphi_{( x_t)}^\oplus}( x_1,\cdots, x_t) + x_t + 1$ is a linear combination of the variables $x_i$ present 
in $\varphi_{(x_t)}^\oplus$, and $\len (h_{t,\varphi_{( x_t)}^\oplus}( x_1,\cdots, x_t) + x_t + 1) = {\cal O}(2^\Delta).$ 
See (\ref{Delta}).
The same is true for $h_{t,\varphi_{(x_t)}^\ominus}(x_1,\cdots,x_t) + x_t + 1$. Therefore, the computation of 
\[ g_t(x_1,\cdots,x_t) \stackrel{def}{=} [ h_{t,\varphi_{(x_t)}^\oplus}(x_1,\cdots,x_t) + x_t + 1 ] \; \cdot \; [ h_{t,\varphi_{(x_t)}^\ominus}(x_1,\cdots,x_t) + x_t + 1 ] \]
is ${\cal O}(2^{\; 2 \; \Delta})$ for all $t$ and $g_t(x_1,\cdots,x_t)$ is a linear combination of the variables $x_i$ present in 
$\varphi_{(x_t)}^\oplus(x_1,\cdots,x_t)$ or $\varphi_{(x_t)}^\ominus(x_1,\cdots,x_t)$, that is $V(x_t)$ as defined in (\ref{V_t}).  \\
\raisebox{.5pt}{\textcircled{\raisebox{-.9pt} {2}}}    
Let us consider $g_t(x_1,\cdots,x_t)$ :
\begin{eqnarray*}
g_t(x_1,\cdots,x_t) = \sum_{(\delta_1,\cdots,\delta_t) \in \{0,1\}^t} c^t_{(\delta_1,\cdots,\delta_t)} \; \; x_1^{\delta_1} 
\cdots x_t^{\delta_t}
\end{eqnarray*}
We have just seen that the computation of all the coefficients $c^t_{(\delta_1,\cdots,\delta_t)}$ is ${\cal O}(2^{\; 2 \; \Delta})$.
Let us note that $g_t(s_1, \cdots, s_t)$ corresponds to the total of the terms in $g_t(x_1,\cdots,x_t)$ not including 
$x_j$ such that $s_j = 0$.  For example, if $g_3(x_1,x_2,x_3) = 1 + x_1 + 3 \; x_1 x_2 + 7 \; x_1 x_2 x_3$,  we get
$g_3(1,1,0) = 5 \stackrel{mod \;2}{=} 1$.   It takes ${\cal O}(2^{\Delta})$ to get this result, as $g_t(x_1,\cdots,x_t)$
is a linear combination of ${\cal O}(\Delta)$ variables, as explained in (\ref{property6}).
Therefore, we have access, in polynomial time, to any of the coefficients $c^t_{(\delta_1, \cdots, \delta_t)}$ and to any
value for $g_t(s_1, \cdots, s_t)$. \\[12pt]
\raisebox{.5pt}{\textcircled{\raisebox{-.9pt} {3}}}    
We have now to compute the coefficient of $\displaystyle ( \prod_{t=1}^n x_t )$ in 
$\displaystyle \prod_{t=1}^n g_t(x_1, \cdots, x_t)$.
But the entire computation of $\displaystyle \prod_{t=1}^n g_t(x_1, \cdots, x_t)$ is ${\cal O}(2^{\;n\;\Delta})$ as each factor has 
${\cal O}(2^\Delta)$ terms. \\
\begin{definition} \mbox{}\\
\begin{itemize}
    \item Let us define the sub-product : $G^{\;[i_1:i_2]}(x_1,\cdots,x_{i_2}) \equiv \displaystyle \prod_{j=i_1}^{i_2} g_j(x_1, \cdots, x_j)$ \\
    \item And $C^{\;[i_1:i_2]}_{(\delta_1, \cdots, \delta_{i_2})}$ as
    $\displaystyle \sum_{(\delta_1,\cdots,\delta_{i_2}) \in \{0,1\}^{i_2}} \hspace{-12pt} C^{\;[i_1:i_2]}_{(\delta_1, \cdots, \delta_{i_2})} \; x_1^{\delta_1} \cdots x_{i_2}^{\delta_{i_2}} \;
    = \; G^{\;[i_1:i_2]}(x_1, \cdots, x_{i_2}) $ \\
    \item[] What we are looking for is $C^{\;[1:n]}_{(1,\cdots,1)}$.\\
\end{itemize}
\end{definition}
\begin{lemma}  \label{lemma1}
\begin{eqnarray*}
C^{\;[1:n]}_{(\delta_1,\cdots,\delta_n)} 
& = &    \sum_{\xi_1=0}^{\delta_1} \;\; \sum_{\zeta_1=\delta_1-\xi_1}^{\delta_1} \cdots 
  \sum_{\xi_{n-1}=0}^{\delta_{n-1}} \;\; \sum_{\zeta_{n-1}=\delta_{n-1}-\xi_{n-1}}^{\delta_{n-1}}
\; \; c^{\;n}_{(\xi_1,\cdots,\xi_{n-1},\delta_n)}  
 \; C^{\;[1:n-1]}_{(\zeta_1,\cdots,\zeta_{n-1})} \\
C^{\;[1:n]}_{(1,\cdots,1)} & = &  \hspace{-24pt}  \sum_{(\xi_1,\cdots,\xi_{n-1}) \in \{0,1\}^{n-1}} 
\; \left( c^{\;n}_{(\xi_1,\cdots,\xi_{n-1},1)} \;\; \sum_{\zeta_1=1-\xi_1}^{1} \cdots 
  \sum_{\zeta_{n-1}=1-\xi_{n-1}}^{1}
 \; C^{\;[1:n-1]}_{(\zeta_1,\cdots,\zeta_{n-1})} \right) \\
\end{eqnarray*}
\mbox{}\\
{\em Proof of the lemma : }  
\begin{eqnarray*} 
G^{\;[1:n]}(x_1, \cdots, x_{n}) & = &  g^{\;n}(x_1,\cdots,x_{n}) \cdot G^{\;[1:n-1]}(x_1,\cdots,x_{n-1})  \\
& &   \\
&  = & \left(  \sum_{(\xi_1,\cdots,\xi_{n}) \in \{0,1\}^{n}} 
c^{\;n}_{(\xi_1,\cdots,\xi_{n})} \; \; x_1^{\xi_1} 
\cdots x_{n}^{\xi_{n}} \right) \\
&  &  \hspace{24pt} \times
\left( \sum_{(\zeta_1,\cdots,\zeta_{n-1}) \in \{0,1\}^{n-1}} C^{\;[1:n-1]}_{(\zeta_1,\cdots,\zeta_{n-1})} \; \; x_1^{\zeta_1} \cdots x_{n-1}^{\zeta_{n-1}} 
\right) \\
& &\\
& &\hspace{-93pt} = \sum_{(\delta_1,\cdots,\delta_{n}) \in \{0,1\}^{n}} \left( \sum_{\xi_1=0}^{\delta_1} \;\; \sum_{\zeta_1=\delta_1 - \xi_1}^{\delta_1} 
\cdots \sum_{\xi_{n-1}=0}^{\delta_{n-1}} \;\; \sum_{\zeta_{n-1} = \delta_{n-1} - \xi_{n-1}}^{\delta_{n-1}}
\hspace{-24pt} c^{\;n}_{(\xi_1,\cdots,\xi_{n-1},\delta_{n})}  
 \; C^{\;[1:n-1]}_{(\zeta_1,\cdots,\zeta_{n-1})} \right) \; \; x_1^{\delta_1} \cdots x_{n}^{\delta_{n}}  \\
\end{eqnarray*}
The coefficient $C^{\;[1:n]}_{(1,\cdots,1)}$ is given by the sum of the terms where 
$x_j^{\;\xi_j} x_j^{\;\zeta_j} = x_j$, i.e. $(\zeta_j,\xi_j) \in \{(0,1),(1,0)(1,1)\}$ as $x^2 \stackrel{mod  \; 2}{=} x$, for
$j \leq i$. \\
\mbox{} \hfill $\blacksquare$
\end{lemma}
\raisebox{.5pt}{\textcircled{\raisebox{-.9pt} {4}}}    
Therefore,
\begin{eqnarray*}
C^{\;[1:2]}_{(1,1)} = & & c^2_{(0,1)} c^1_{(1)} + c^2_{(1,1)} (c^1_{(0)} + c^1_{(1)}) \\
C^{\;[1:2]}_{(0,1)} = & & c^2_{(0,1)} c^1_{(0)}  \\
C^{\;[1:2]}_{(1,0)} = & & c^2_{(0,0)} c^1_{(1)} + c^2_{(1,0)} (c^1_{(0)} + c^1_{(1)}) \\
C^{\;[1:2]}_{(0,0)} = & & c^2_{(0,0)} c^1_{(0)}  \\
\mbox{} & &\\
C^{\;[1:3]}_{(1,1,1)}  = & & c^3_{(0,0,1)} C^{\;[1:2]}_{(1,1)}  \\
& + & c^3_{(0,1,1)} \left( C^{\;[1:2]}_{(1,0)} + C^{\;[1:2]}_{(1,1)} \right) 
 +  c^3_{(1,0,1)} \left( C^{\;[1:2]}_{(0,1)} + C^{\;[1:2]}_{(1,1)} \right) \\
& + & c^3_{(1,1,1)} \left( C^{\;[1:2]}_{(0,0)} + C^{\;[1:2]}_{(1,0)} + 
C^{\;[1:2]}_{(0,1)} + C^{\;[1:2]}_{(1,1)} \right)  \\
C^{\;[1:3]}_{(0,1,1)}  = & & c^3_{(0,0,1)} C^{\;[1:2]}_{(0,1)} +  c^3_{(0,1,1)} \left( C^{\;[1:2]}_{(0,0)} + C^{\;[1:2]}_{(0,1)} \right) \\
C^{\;[1:3]}_{(0,0,1)}  = & & c^3_{(0,0,1)} C^{\;[1:2]}_{(0,0)} \\
\mbox{and so on } \cdots  & &
\end{eqnarray*}
Each $c^i_{(\delta_1,\cdots,\delta_{i-1},\delta_i)}$ is multiplied by $2^{(\delta_1+ \cdots + \delta_{i-1})} \;$ terms $C^{\;[1:i-1]}_{(\cdot)}$, as every ``1" in \linebreak $c^i_{(\delta_1,\cdots,\delta_{i-1},\delta_i)}$, except for $\delta_i$, is replaced by a ``0" and a ``1" in the factors $C^{\;[1:i-1]}_{(\delta_1,\cdots,\delta_{i-1})}$.
So, to get $C^{\;[1:i]}_{(\delta_1,\cdots,\delta_i)}$, one should compute  :
\begin{eqnarray}  \label{proof1} 
3^{(\delta_1+ \cdots + \delta_{i-1})} &=& \sum_{j=0}^{\delta_1 + \cdots + \delta_{i-1}} \binom{\delta_1+ \cdots + \delta_{i-1}}{i}
 2^{j}  \;\;\; \mbox{ terms,} \\ \nonumber
\end{eqnarray}
and to get all the $C^{\;[1:i]}_{(\cdot)}$, this needs : 
\begin{eqnarray} \label{proof2}
2 \cdot 4^{(i-1)} &=& 2 \cdot \; \sum_{j=0}^{i-1} \binom{i-1}{j} 3^{j}  \;\;\; \mbox{ terms, considering }
\delta_i \in \{0,1\}. \\ \nonumber
\end{eqnarray}
\raisebox{.5pt}{\textcircled{\raisebox{-.9pt} {5}}}    
{\bf But:} 
\begin{enumerate}
    \item A very limited number of $c^i_{(\cdot)}$ are non zero in a hard 3-CNF-SAT problem.
    \item $g_i(x_1,\cdots,x_i)$ is a linear combination of $V(x_i)$ variables, with $\# V(x_i) = {\cal O}(\Delta)$.
    \item So, the maximum number of non-zero $c^i_{(\cdot)}$ is ${\cal O}(2^\Delta)$ in $g_i(x_1,\cdots,x_i)$.
    \item And the maximum total number of non-zero $c^i_{(\cdot)}$ for all $1 \leq i \leq n$ is ${\cal O}(n \cdot 2^\Delta)$.
\end{enumerate}
\mbox{}\\
\begin{lemma} \label{lemma2}
The complexity to compute $C^{\;[1:i]}_{(\delta_1,\cdots,\delta_i)}$ from $C^{\;[1:i-1]}_{(\delta_1,\cdots,\delta_{i-1})}$ is ${\cal O}(3^\Delta)$
\end{lemma}
\mbox{}\\
{\em Proof of the lemma : }  \\
Without loss of generality, one can suppose that $V(x_i) = \{x_{i - k \Delta}, \cdots, x_i \}$, as the number of variables in $g_i(x_1,\cdots,x_i)$ is ${\cal O}(\Delta)$. Thus, 
\[ c^i_{(\delta_1,\cdots,\delta_i)} = 0 \;\;\; \mbox{if } \;\; \; \delta_j = 1 \;\;\; \mbox{ for some } j < i- k \Delta \]
The formula (\ref{proof1}) to get $C^{\;[1:i]}_{(\delta_1,\cdots,\delta_i)}$ should be modified :
\begin{eqnarray}  \label{proof3} 
3^{(\delta_{i - k \Delta}+ \cdots + \delta_{i-1})} &=& \sum_{j=0}^{\delta_{i-k\Delta} + \cdots + \delta_{i-1}} \binom{\delta_{i-k\Delta}+ \cdots + \delta_{i-1}}{i}
 2^{j}   \\ \nonumber 
& & \\ \nonumber 
& =& {\cal O}(3^{\Delta}) \mbox{ as } \;\;
(\delta_{i-k\Delta}+ \cdots + \delta_{i-1}) \leq k \Delta 
\end{eqnarray}
\mbox{} \hfill $\blacksquare$ \\
Unfortunately, the formula (\ref{proof2}) to get all $C^{\;[1:i]}_{(\delta_1,\cdots,\delta_i)}$ is still ${\cal O}(2^i)$. 
\begin{eqnarray} \label{lemma3}
2^{(i-k\Delta)} \cdot 4^{(k \Delta)} &=& 2^{i-k\Delta} 
\cdot \; \sum_{j=0}^{k \Delta} \binom{k \Delta}{j} 3^{j}  \\\nonumber
& & \mbox{ as } \;\;
(\delta_{i-k\Delta}+ \cdots + \delta_{i-1}) \leq k \Delta \\ \nonumber
& & \mbox{ and } \;\; (\delta_1,\cdots,\delta_{i- k\Delta -1)}) \in \{0,1\}^{(i-k\Delta-1)} \; \; \mbox{ and } \;\; (\delta_i) \in \{0,1\}  \\ \nonumber
& = & {\cal O}(2^i)
\end{eqnarray}
So, the complexity to get all $C^{\;[1:n]}_{(\delta_1,\cdots,\delta_n)}$ from $C^{\;[1:n-1]}_{(\delta_1,\cdots,\delta_{n-1})}$ is ${\cal O}(2^n)$.\\[12pt]
\raisebox{.5pt}{\textcircled{\raisebox{-.9pt} {6}}}    
{\bf BUT (this is the turning point of the proof)
if we only need to know $C^{\;[1:n]}_{(1,\cdots,1)}$ :} \\
\begin{enumerate}
    \item The complexity to {\em compute} $C^{\;[1:n]}_{(1,\cdots,1)}$ from $C^{\;[1:n-1]}_{(\delta_1,\cdots,\delta_{n-1})}$ is ${\cal O}(3^{\Delta})$, by {\em Lemma \ref{lemma2}} \\
    \item To compute $C^{\;[1:n]}_{(1,\cdots,1)}$, one need to {\em know} ${\cal O}(2^{\Delta})$ terms $C^{\;[1:n-1]}_{(\delta_1,\cdots,\delta_{n-1})}$.  Indeed, if $V(x_n) = \{x_{n-k\Delta},\cdots, x_n\}$ :
    \begin{eqnarray*} [ \delta_i = 1 \; \mbox{ for } \;  i < n-k\Delta] \;\;& \Rightarrow &\;\;  c^n_{(\delta_1,\cdots,\delta_{n-k\Delta},\cdots,\delta_n)} = 0 \\ 
    & \Downarrow & \\
    \mbox{Only }\;\; C^{\;[1:n-1]}_{(0,\cdots,0,\delta_{n-k\Delta},\cdots,\delta_{n-1})} \;\; &\mbox{ with } & \; \; 
    (\delta_{n-k\Delta},\cdots,\delta_n) \in \{0,1\}^{k\Delta} \\[9pt]
    \mbox{can occur in } C^{\;[1:n]}_{(1,\cdots,1)}, \mbox{ depending} & \mbox{of the}& \mbox{values for } \; 
c^{n}_{(0,\cdots,0,\delta_{n-k\Delta},\cdots,\delta_n)}.
    \end{eqnarray*}
\item {\bf AND} to {\em compute} these ${\cal O}(2^{\Delta})$ terms $C^{\;[1:n-1]}_{(0,\cdots,0,\delta_{n-k\Delta},\cdots,\delta_{n-1})}$, the complexity given by (\ref{lemma3}) is not exponential but ${\cal O}(4^\Delta)$ : 
    \begin{eqnarray} \label{lemma4}
2 \cdot 4^{(k \Delta - 1)} &=&  
2 \cdot \; \sum_{j=0}^{k \Delta - 1} \binom{k \Delta -1}{j} 3^{j}  \\\nonumber
& & \mbox{ as } \;\;
(\delta_{n-k\Delta}+ \cdots + \delta_{n-2}) \leq k \Delta -1 \\ \nonumber
& & \mbox{ and } \;\; (\delta_1,\cdots,\delta_{i- k\Delta -1)}) = (0,\cdots,0) \; \; \mbox{ and } \;\; (\delta_{n-1}) \in \{0,1\}  \\ \nonumber
& = & {\cal O}(4^\Delta) \\ \nonumber
& & \mbox{ (instead of $ {\cal O}(3^\Delta)$ for the previous step for only one calculation)}
\end{eqnarray}
\item And we still need to {\em know only} ${\cal O}(2^{\Delta})$ terms $C^{\;[1:n-2]}_{(\delta_1,\cdots,\delta_{n-2})}$ to compute the ${\cal O}(2^{\Delta})$ terms $C^{\;[1:n-1]}_{(0,\cdots,0,\delta_{n-k\Delta},\cdots,\delta_{n-1})}$.  Indeed, $g^{n-1}(x_1,\cdots,x_{n-1})$ is a linear combination of ${\cal O}(2^{\Delta})$ variables, the ones in $V(x_{n-1}).$
    \begin{eqnarray*} 
\mbox{So, }\;\;    \delta_i = 1 \; \mbox{ for } \;  i \;:\; x_i \not \in V(x_{n-1})  \;\;& \Rightarrow &\;\;  c^{n-1}_{(\delta_1,\cdots,\delta_{n-1})} = 0 \\ 
    & \Downarrow & \\
    \mbox{Only }\;\; C^{\;[1:n-2]}_{(\delta_1,\cdots,\delta_{n-2})} \;\; &\mbox{ with } & \; \; 
    \delta_{i}=0 \; \mbox{ for } \;  i \;:\; x_i \not \in V(x_{n-1})  \\[9pt]
    \mbox{can occur in } C^{\;[1:n-1]}_{(0,\cdots,0,\delta_{n-k\Delta},\cdots,\delta_{n-1})},& & 
    \hspace{-30pt} \mbox{ depending of the values for } \; 
c^{n-1}_{(\delta_1,\cdots,\delta_{n-1})}.
    \end{eqnarray*} 
Thus, at most ${\cal O}(2^{\Delta})$ terms $C^{\;[1:n-2]}_{(\delta_1,\cdots,\delta_{n-2})}$ will be needed.\\  
\item This step is similar to step (3) : the complexity to {\em compute} these ${\cal O}(2^{\Delta})$ terms $C^{\;[1:n-2]}_{(\delta_{1},\cdots,\delta_{n-2})}$ with 
    $\delta_{i}=0 \; \mbox{ for } \;  i \;:\; x_i \not \in V(x_{n-2})$, is once again $ {\cal O}(4^\Delta)$ :
    \begin{eqnarray*} 
2 \cdot 4^{(k \Delta - 1)} &=&  
2 \cdot \; \sum_{j=0}^{k \Delta - 1} \binom{k \Delta -1}{j} 3^{j}  \\
& & \mbox{ as } \;\;
(\delta_{1}+ \cdots + \delta_{n-3}) \leq [\# V(x_{n-2})-1] = k\Delta-1 \\ 
& & \mbox{ and } \;\;\delta_{i}=0 \; \mbox{ for } \;  i \;:\; x_i \not \in V(x_{n-2}) \; \; \mbox{ and } \;\; (\delta_{n-2}) \in \{0,1\}  \\ 
& = & {\cal O}(4^\Delta) \\
\end{eqnarray*} 
\item This step is similar to step (4) : the same arguments with $V(x_{n-2})$ yields to the same conclusion, that one need to know at most ${\cal O}(2^{\Delta})$ terms $C^{\;[1:n-3]}_{(\delta_1,\cdots,\delta_{n-3})}$ to compute the previous step.\\
\item And so on, till $g^1(x_1)$.
 \end{enumerate}
 \mbox{}\\
{\bf In conclusion, the complexity to compute $C^{\;[1:n]}_{(1,\cdots,1)}$ is the sum of the computing complexity of each even step : }
\begin{eqnarray}
{\cal O}(3^\Delta) + (n-1) {\cal O}(4^\Delta) = {\cal O}(n) \label{main_result}
\end{eqnarray}
\mbox{} \hfill $\blacksquare$
\mbox{}\\
\mbox{}\\
\mbox{}\\
{\bf Example :} Let us consider this situation :
\begin{itemize}
    \item $V(x_i) = \{x_{i-2},x_{i-1},x_i\}$ for $3 \leq i \leq n$
    \item $g^i(x_1,\cdots,x_i) = 1 + x_{i-2} + x_{i-1} + x_i + x_{i-2}\; x_{i-1} + x_{i-2}\; x_i
+ x_{i-1}\; x_i + x_{i-2}\; x_{i-1}\; x_i$
\item[$\Rightarrow$] $c^i_{(\delta_1,\cdots,\delta_i)} = 0$ if $\delta_j \neq 0$ for $j<i-2$ and 
$c^i_{(0,\cdots,0,\delta_{i-2},\delta_{i-1},\delta_i)} = 1$ otherwise
\end{itemize}
Then,  
\begin{eqnarray*}
C^{\;[3:n]}_{(1,\cdots,1,1,1)}& = &\;\;\;c^n_{(0,\cdots,0,\stackrel{\delta_{n-2}}{0},\stackrel{\delta_{n-1}}{0},\stackrel{\delta_{n}}{1})} \big(C^{\;[3:n-1]}_{(1,\cdots,1,\stackrel{\delta_{n-2}}{1},\stackrel{\delta_{n-1}}{1})}\big) \\
& & + \;\; c^n_{(0,\cdots,0,0,1,1)} \big(C^{\;[3:n-1]}_{(1,\cdots,1,1,0)}+C^{\;[3:n-1]}_{(1,\cdots,1,1,1)}\big) \\
& & + \;\;c^n_{(0,\cdots,0,1,0,1)} \big(C^{\;[3:n-1]}_{(1,\cdots,1,0,1)}+C^{\;[3:n-1]}_{(1,\cdots,1,1,1)}\big) \\
& & + \;\; c^n_{(0,\cdots,0,1,1,1)} \big(C^{\;[3:n-1]}_{(1,\cdots,1,0,0)}+C^{\;[3:n-1]}_{(1,\cdots,1,0,1)} +
C^{\;[3:n-1]}_{(1,\cdots,1,1,0)}+C^{\;[3:n-1]}_{(1,\cdots,1,1,1)}\big) \\[12pt]
& = & 4 \; C^{\;[3:n-1]}_{(1,\cdots,1,1,1)} + 2  \; C^{\;[3:n-1]}_{(1,\cdots,1,0,1)} + 2 \;  C^{\;[3:n-1]}_{(1,\cdots,1,1,0)} +  C^{\;[3:n-1]}_{(1,\cdots,1,0,0)} \\
& \stackrel{mod\;2}{=} &  C^{\;[3:n-1]}_{(1,\cdots,1,0,0)} \\
& &\\
C^{\;[3:n-1]}_{(1,\cdots,\stackrel{\delta_{n-3}}{1},\stackrel{\delta_{n-2}}{0},\stackrel{\delta_{n-1}}{0})} & = &\;\;\;c^{n-1}_{(0,\cdots,0,\stackrel{\delta_{n-3}}{0},\stackrel{\delta_{n-2}}{0},\stackrel{\delta_{n-1}}{0})} \big(C^{\;[3:n-2]}_{(1,\cdots,1,\stackrel{\delta_{n-3}}{1},\stackrel{\delta_{n-2}}{1})}\big) \\
& & + \;\;c^{n-1}_{(0,\cdots,0,\stackrel{\delta_{n-3}}{1},\stackrel{\delta_{n-2}}{0},\stackrel{\delta_{n-1}}{0})} \big(C^{\;[3:n-2]}_{(1,\cdots,1,\stackrel{\delta_{n-3}}{0},\stackrel{\delta_{n-2}}{1})} + 
C^{\;[3:n-2]}_{(1,\cdots,1,\stackrel{\delta_{n-3}}{1},\stackrel{\delta_{n-2}}{1})} \big)\\[12pt]
& = & 2 \; C^{\;[3:n-2]}_{(1,\cdots,1,1,1)}  +  C^{\;[3:n-2]}_{(1,\cdots,1,0,0)} \\
& \stackrel{mod\;2}{=} &  C^{\;[3:n-2]}_{(1,\cdots,1,0,0)} \\
& &\\
C^{\;[3:n-2]}_{(1,\cdots,\stackrel{\delta_{n-4}}{1},\stackrel{\delta_{n-3}}{0},\stackrel{\delta_{n-2}}{0})} & = &\;\;\;c^{n-2}_{(0,\cdots,0,\stackrel{\delta_{n-4}}{0},\stackrel{\delta_{n-3}}{0},\stackrel{\delta_{n-2}}{0})} \big(C^{\;[3:n-3]}_{(1,\cdots,1,\stackrel{\delta_{n-4}}{1},\stackrel{\delta_{n-3}}{1})}\big) \\
& & + \;\;c^{n-2}_{(0,\cdots,0,\stackrel{\delta_{n-4}}{1},\stackrel{\delta_{n-3}}{0},\stackrel{\delta_{n-2}}{0})} \big(C^{\;[3:n-3]}_{(1,\cdots,1,\stackrel{\delta_{n-4}}{0},\stackrel{\delta_{n-3}}{1})} + 
C^{\;[3:n-3]}_{(1,\cdots,1,\stackrel{\delta_{n-4}}{1},\stackrel{\delta_{n-3}}{1})} \big)\\[12pt]
& = & 2 \; C^{\;[3:n-3]}_{(1,\cdots,1,1,1)}  +  C^{\;[3:n-3]}_{(1,\cdots,1,0,0)} \\
& \stackrel{mod\;2}{=} &  C^{\;[3:n-3]}_{(1,\cdots,1,0,0)} \\
& &\\
\Rightarrow \;\;\; C^{\;[3:n]}_{(1,\cdots,1,1,1)} & = & C^{\;[3:3]}_{(1,0,0)} = c^3_{(1,0,0)} = 1 \\
\end{eqnarray*}
\begin{corollary} \label{corollary1}
\mbox{}\\
{\bf A 3-CNF-SAT problem $\varphi$ with $\Delta  = {\cal O}(1)$ such that [$C^{\;[1:n]}_{(1,\cdots,1)} = 1 \Leftrightarrow  {\cal S}_\varphi \neq \emptyset$] is \classP.}
\end{corollary}
\begin{proof}
As the computation of $C^{\;[1:n]}_{(1,\cdots,1)}$ is $ {\cal O}(n)$ when  $\Delta  = {\cal O}(1)$, 
the satisfiability problem of $\varphi$ is $ {\cal O}(n)$.
\end{proof}
\mbox{}\\
\begin{corollary}
\mbox{}\\
{\bf A 3-CNF-SAT problem $\varphi$ with $\Delta  = {\cal O}(1)$ and $\# {\cal S}_\varphi = 2\;k + 1 \; (k\in \mathbbm{N}) $ is \classP.}
\end{corollary}
\begin{proof}
From (\ref{indicator4}), we get :
\begin{eqnarray*}
\mathbbm{1}_{{\cal S}_\varphi}(x_1, \cdots, x_n) & = & \sum_{(s_1, \cdots,s_n) \;\in\; {\cal S}_\varphi} \;  \prod_{i=1}^n 
x_i^{s_i} (x_i + 1)^{(s_i + 1)} \;\;\; \mbox{\tiny \em (mod 2)} \\
&\stackrel{(\ref{satis1})}{=} & \;\;\; \sum_{j=1}^{2k+1} \; 
\left( \prod_{t = 1}^n x_t   + {\cal E}_j(x_1,\cdots,x_n) \right) \\
& = & (2k + 1) \left( \prod_{t = 1}^n x_t \right) \;\; +\;\; \sum_{j=1}^{2k+1} \;
{\cal E}_j(x_1,\cdots,x_n) \\
& \stackrel{mod\;2}{=} & \left( \prod_{t = 1}^n x_t \right) \;\; +\;\;  {\cal E}(x_1,\cdots,x_n) 
\end{eqnarray*} 
Therefore, the hypotheses of {\em Corollary \ref{corollary1}} are satisfied, and the 3-CNF-SAT problem is \classP \\ 
\end{proof}
\subsection{Complexity theorems for 3-CNF-SAT problems with $\# {\cal S}_\varphi \leq 2^k$}
\noindent We have found a polynomial algorithm to solve the satisfiability of any 3-CNF-SAT problem, 
assuming only zero or an {\em odd} number of solutions can occur.
What about 3-CNF-SAT problems with a possible {\em even} number of solutions ? \\[18pt]
\raisebox{.5pt}{\textcircled{\raisebox{-.9pt} {1}}}     
Let us begin with $\varphi$ such that $\# {\cal S}_\varphi$ is zero or $2$ and 
${\cal S}_\varphi = \{(s_1,\cdots,s_n),(s'_1,\cdots,s'_n)\}$ or $\emptyset$.\\[12pt]
$\bullet \;\;\mbox{ Suppose } \exists \; ! j \;\in \{1,\cdots,n\}$ such that $s_j \neq s'_j$. 
\begin{eqnarray*}
\mathbbm{1}_{{\cal S}_\varphi}(x_1,\cdots,x_n) &=& \prod_{i=1}^n \; x_i^{s_i}\; (1+x_i)^{1+s_i}
\; + \; \prod_{i=1}^n \; x_i^{s'_i}\; (1+x_i)^{1+s'_i}\\
& = & \prod_{i \; \neq j} \; x_i^{s_i}\; (1+x_i)^{1+s_i}
\; \big( \;  x_j^{s_j}\; (1+x_j)^{1+s_j} \; + \;
 x_j^{s'_j}\; (1+x_j)^{1+s'_j} \; \big) \\
& = & \prod_{i \; \neq j} \; x_i^{s_i}\; (1+x_i)^{1+s_i}
\; \big( \;  x_j\; + \; (1+x_j) \; \big) \\
& = & \prod_{i \; \neq j} \; x_i^{s_i}\; (1+x_i)^{1+s_i} \\
& = & \big(\;\prod_{\stackrel{i=1}{i \; \neq j}}^n \; x_i\;\big) + {\cal E}(x_1,\cdots,x_{j-1},x_{j+1},\cdots,x_n) \\
\Rightarrow \;\; \varphi \;\mbox{ is satisfiable } & \Leftrightarrow & \exists \; C^{\;[1:n]}_{(\delta_1,\cdots,\delta_n)} \; \neq 0 \; \mbox{ with $\sum$} \delta_i = n-1 \; \mbox{ [one } \delta_i \mbox{ equals 0, i.e. } \delta_j=0].
\end{eqnarray*}
$\bullet\;\;$ Let us consider the general situation : $I^{(1)}= \{i \;:\; s_i = 0\}$ and $I^{(2)}= \{i \;:\; s'_i = 0\}$. Let $I^{(1)} = \{i^{(1)}_1,\cdots,i^{(1)}_{m^{(1)}}\} \;,\; I^{(2)} = \{i^{(2)}_1,\cdots,i^{(2)}_{m^{(2)}}\}$ and $\mu (j) \stackrel{def}{=} 
\#\{ l : j \;\in \; I^{(l)} \}$.
\begin{eqnarray*}
\mathbbm{1}_{{\cal S}_\varphi}(x_1,\cdots,x_n) &=&  \prod_{i \; \not \in I^{(1)}} x_i\; \prod_{i \; \in I^{(1)}} (1+x_i)
 \; + \;  \prod_{i \; \not \in I^{(2)}} x_i \prod_{i \; \in I^{(2)}} (1+x_i) \\[18pt]
& & \hspace{-20pt}
\left[ \begin{array}{ll}
\displaystyle \;\; \prod_{i\;\not \in I^{(1)}}  x_i \;
\prod_{i \; \in I^{(1)}}  (1+x_i)&\displaystyle  = ( \prod_{i\;\not \in I^{(1)}}  x_i )\;
(\sum_{(\delta_{1},\cdots,\delta_{m^{(1)}}) \;\in \{0,1\}^{m^{(1)}}} \;\;\prod_{j=1}^{m^{(1)}} x_{i_j}^{\delta_j}) \\[18pt]
&\displaystyle = \prod_{i=1}^n x_{i} + \prod_{i=1\;,\;i \neq i^{(1)}_1}^n x_{i} +
\prod_{i\neq i^{(1)}_2} x_{i} +  \cdots \\[18pt]
\displaystyle \;\; \prod_{i\;\not \in I^{(2)}}  x_i \;
\prod_{i \; \in I^{(2)}}  (1+x_i)&\displaystyle  =  \prod_{i=1}^n x_{i} + \prod_{i \neq i^{(2)}_1} x_{i} +
\prod_{i\neq i^{(2)}_2} x_{i} +  \cdots 
\end{array} \right] \\[18pt]
& \stackrel{mod\;2}{=} & \sum_{j\;:\;\mu(j)=1}  \big(\;\prod_{\stackrel{i=1}{i \; \neq j}}^n \; x_i\;\big)  + {\cal E}(x_1,\cdots,x_n) \;\; \;\; \big[ (\# {\cal S}_\varphi = 2) \; \Rightarrow \{j\;:\; \mu(j)=1\} \neq \emptyset \big]
\\[18pt]
\varphi \;\mbox{ is satisfiable } & \Leftrightarrow & \exists \; C^{\;[1:n]}_{(\delta_1,\cdots,\delta_n)} \; \neq 0 \; \mbox{ with $\sum$} \delta_j = n-1 \; \mbox{[any } \delta_{j}=0\mbox{ with } \mu(j)=1].
\end{eqnarray*}
{\em So, the unsatisfiability of any 3-CNF-SAT $\varphi$ with maximum 2 solutions will be proved only if one gets :}
\[ C^{\;[1:n]}_{(1,\cdots,1)} = 0 \; \mbox{ and } C^{\;[1:n]}_{(\delta_1,\cdots,\delta_n)} = 0 \; \mbox{ for the $\binom{n}{1}$ situations where }\sum_{j=1}^n \delta_j = n-1 \] 
 Using the same arguments as to get (\ref{main_result}), the complexity to compute any $C^{\;[1:n]}_{(\delta_1,\cdots,\delta_n)} \linebreak \mbox{with $\sum$} \delta_i = n-1 $ is ${\cal O}(n)$.  So the general complexity for any 3-CNF-SAT problem with $\# {\cal S}_\varphi \leq 2$ is ${\cal O}(n) +
 \binom{n}{1} {\cal O}(n) = {\cal O}(n^2)$.\\[24pt]
\raisebox{.5pt}{\textcircled{\raisebox{-.9pt} {2}}}     
Let us consider $\varphi$ such that $\# {\cal S}_\varphi = K = 2^k$ and 
${\cal S}_\varphi = \{(s^{(1)}_1,\cdots,s^{(1)}_n),\cdots,(s^{(K)}_1,\cdots,s^{(K)}_n)\}$.\\[12pt]
\begin{tabular}{ll}
\mbox{Let }& $I^{(l)}  = \{i \;:\; s^{(l)}_i = 0\} \mbox{ for } 1 \leq l \leq K $\\[9pt]
 & $I^{(l)}  =  \{i^{(l)}_1,\cdots,i^{(l)}_{m^{(l)}}\} $ \\[9pt]
\mbox{and }&  $\mu (\{j_1,\cdots,j_p\})  \stackrel{def}{=} 
\#\{ l : \{j_1,\cdots,j_p\} \;\subseteq \; I^{(l)} \}.$ 
\end{tabular}
\begin{eqnarray*}
\mathbbm{1}_{{\cal S}_\varphi}(x_1,\cdots,x_n) &=&  \sum_{l=1}^K \big( \prod_{i \; \not \in I^{(l)}} x_i\; \prod_{i \; \in I^{(l)}} (1+x_i) \big) \\[18pt]
& & \hspace{-35pt}
\left[ \begin{array}{ll}
\displaystyle \;\; \prod_{i\;\not \in I^{(l)}}  x_i \;
\prod_{i \; \in I^{(l)}}  (1+x_i)&\displaystyle  = ( \prod_{i\;\not \in I^{(l)}}  x_i )\;
(\sum_{(\delta_{1},\cdots,\delta_{m^{(l)}}) \;\in \{0,1\}^{m^{(l)}}} \;\;\prod_{j=1}^{m^{(l)}} x_{i_j}^{\delta_j}) \\[18pt]
&\displaystyle = \prod_{i=1}^n x_{i} + \sum_{j_1\;\in I^{(l)}} \;\; \prod_{i \neq j_1} x_{i} +
\sum_{\{j_1,j_2\}\;\subseteq I^{(l)}} \;\; \prod_{i \not \in \{j_1,j_2\}} x_{i} \\ [18pt]
&\displaystyle \;\;\; +\;\; \cdots  \sum_{\{j_1,\cdots,j_{m^{(l)}}\}\;\subseteq I^{(l)}} \;\; \prod_{i \not \in \{j_1,\cdots,j_{m^{(l)}}\}} x_{i} 
\end{array} \right] \\[18pt]
& \stackrel{mod\;2}{=} &  \sum_{j_1\;:\; \mu(\{j_1\}) = 1 } \;\; \prod_{i \neq j_1} x_{i} +
\sum_{\{j_1,j_2\}\;:\; \mu(\{j_1,j_2\}) = 1 } \;\; \prod_{i \not \in \{j_1,j_2\}} x_{i} 
\;\; +\;\; \cdots \\
& & \;\;\; + \sum_{\{j_1,\cdots,j_p\}\;:\; \mu(\{j_1,\cdots,j_p\}) = 1 }\;\; \prod_{i \not \in \{j_1,\cdots,j_p\}} x_{i} \;+\;\; {\cal E}(x_1,\cdots,x_n) \\[18pt]
& & \hspace{-35pt}
\left[ \begin{array}{l}
\displaystyle \;\; \# {\cal S}_\varphi \leq 2^k \;\;\Rightarrow  \;\; \big\{ \{j_1,\cdots,j_k\}\;:\; 
\mu(\{j_1,\cdots,j_k\}) = 1 \big\} \neq \emptyset \\[18pt]
\mbox{It is possible that : }   \big\{ \{j_1,\cdots,j_p\}\;:\; 
\mu(\{j_1,\cdots,j_p\}) = 1 \big\} = \emptyset \;\mbox{ for } p < k \\
\mbox{For example, } {\cal S}_\varphi  =  \{(1,1,0,0,0),(1,1,1,0,0),(1,1,0,1,0), (1,1,0,0,1),\\
 \hspace{90pt} (1,1,1,1,0),(1,1,1,0,1),(1,1,0,1,1),(1,1,1,1,1)\} \\
 \Rightarrow 
\big\{ \{j_1\}\;:\; \mu(\{j_1\}) = 1 \big\} = \big\{ \{j_1,j_2\}\;:\; 
\mu(\{j_1,j_2\}) = 1 \big\} = \emptyset \\
 \Rightarrow  \big\{ \{j_1,j_2,j_3\}\;:\; 
\mu(\{j_1,j_2,j_3\}) = 1 \big\} \neq \emptyset
\end{array} \right] \\[18pt]
\mbox{Assuming } & & \hspace{-24pt} \# {\cal S}_\varphi \leq 2^k\;:\; \varphi \;\mbox{ is satisfiable }  \Leftrightarrow  \exists \; C^{\;[1:n]}_{(\delta_1,\cdots,\delta_n)} \; \neq 0 \; \mbox{ with $\sum$} \delta_j \leq n-k.
\end{eqnarray*}
{\em So, assuming that our 3-CNF-SAT problems have at most $2^k$ solutions, the unsatisfiability of any such 3-CNF-SAT $\varphi$  will be proved only if one gets :}
\[ C^{\;[1:n]}_{(\delta_1,\cdots,\delta_n)} = 0 \; \mbox{ for the $\displaystyle \sum_{i=0}^k \binom{n}{i}$ situations where }\sum_{j=1}^n \delta_j \leq n-k \] 
 Using the same arguments as to get (\ref{main_result}), the complexity to compute any $C^{\;[1:n]}_{(\delta_1,\cdots,\delta_n)} \linebreak \mbox{with $\sum$} \delta_i \leq n $ is shown to be ${\cal O}(n)$.  \\[18pt]
\textbf{\em In conclusion, on the assumption of a $2^k$ limit for the number of solutions, the general complexity for any 3-CNF-SAT problem with $\Delta = \frac{m}{n} = {\cal O}(1)$ is}
\[\sum_{i=0}^k \binom{n}{i} \; {\cal O}(n) \le (n+1)^k \; {\cal O}(n)
 \;=\; {\cal O}(n^k) \; \mbox{ for large $n$ wrt $k$.} \]
 \mbox{}\\
\textbf{\em Therefore, the ``hard" 3-CNF-SAT problems are in \classP, as their number of solutions 
is limited by $2^k$.  } \\[24pt]

\section{Remarks and further researches}
\noindent It is important to stress that this is not
a heuristic proof. The fact that 
our polynomial algorithm {\bf does not deliver any solution}, but only states whether they exist or not, is
a key issue for downgrading the complexity level in our paper.  This is a very high price to pay, perhaps further
researches could lighten this price.   \\[12pt]
\noindent It is also essential to underline that {\bf this is not a proof that \classNP = \classP}.  It is
a first insight in the complex question of the boundary between \classP and \classNP.  The search of 
a polynomial algorithm for {\em easy} 3-CNF-SAT problems [$\# {\cal S}_\varphi \neq {\cal O}(2^k)$]
is under way, but the main issue seems to be {\em how to distinguish between easy and hard 3-CNF-SAT
problems} and {\em if it is possible to do that in a polynomial complexity.} \\[12pt]
\noindent An algorithm, freely available, exists and was heavily used to check the theoretical 
results of this paper.

%
\nocite{Sipser92}

\addcontentsline{toc}{part}{\mbox{Bibliography}}  
\bibliographystyle{plain}   
\bibliography{mybib}   
\end{document}